\documentclass[pdflatex,sn-basic,Numbered]{sn-jnl}
\usepackage{graphicx,hyperref,enumitem,grffile,parskip}
\usepackage{amssymb,amsmath,amsthm,xcolor,mathrsfs}
\usepackage{tikz,longtable,tcolorbox,array,booktabs,subcaption}
\usepackage{listings}
\usepackage{doi}

\usetikzlibrary{arrows, positioning, shapes.geometric}
\tikzstyle{process} = [ellipse, minimum width=1cm, minimum height=0.6cm, text centered, draw=black]
\tikzstyle{process1} = [ellipse, minimum width=1cm, minimum height=1cm, text centered, draw=white]
\tikzstyle{process2} = [ellipse, minimum width=1.8cm, minimum height=1cm, text centered, draw=blue]
\tikzstyle{decision} = [diamond, aspect = 3, text centered, draw=black]
\tikzstyle{bag} = [align=center]
\tikzstyle{arrow} = [->,>=stealth]

\renewcommand{\AA}{{\mathcal A}}
\newcommand{\BB}{{\mathcal B}}
\newcommand{\CC}{{\mathcal C}}

\newcommand{\GG}{{\mathcal G}}

\newcommand{\II}{{\mathcal I}}
\newcommand{\KK}{{\mathcal K}}

\newcommand{\LL}{{\mathcal L}}

\newcommand{\sS}{{\mathcal S}}

\newcommand{\PH}{H}
\newcommand{\R}{{\mathbf R}}

\newcommand{\biarrow}{\;\smash{{}^{\displaystyle \rightarrow}_{\displaystyle\leftarrow}}\;}

\theoremstyle{thmstyleone}%
\newtheorem{theorem}{Theorem}[section]
\newtheorem{lemma}[theorem]{Lemma}
\newtheorem{definition}[theorem]{Definition}
\newtheorem{corollary}[theorem]{Corollary}

\newtheorem{remark}[theorem]{Remark}

\newcommand{\ignore}[1]{}

\newcommand{\END}{\hfill\mbox{\raggedright$\Diamond$}}
\newcommand{\etal}{\textrm{et al.\,}}

\newcommand{\Matrixc}[1]{\ensuremath{\left[\begin{array}{cccccccccccc|ccc} #1 \end{array}\right]}}

\numberwithin{equation}{section}
\renewcommand{\theequation}{\thesection.\arabic{equation}}
\begin{document}

\title[Article Title]{Automated Classification of Homeostasis Structure in Input-Output Networks}

\author[1]{\fnm{Xinni} \sur{Lin}}

\author[2]{\fnm{Fernando} \sur{Antoneli}}

\author[3]{\fnm{Yangyang} \sur{Wang}}

\affil[1]{\orgdiv{Department of Mathematics}, \orgname{University of Southern California}, \orgaddress{\city{Los Angeles}, \state{California}, \postcode{90007}, \country{USA}}. Email: xinnilin@usc.edu}

\affil[2]{\orgdiv{Escola Paulista de Medicina}, \orgname{Universidade Federal de São Paulo}, \orgaddress{\city{São Paulo}, \state{SP} \postcode{04039-032}, \country{Brazil}}. Email: fernando.antoneli@unifesp.br}

\affil[3]{\orgdiv{Department of Mathematics, Volen National Center for Complex Systems}, \orgname{Brandeis University}, \orgaddress{\city{Waltham}, \state{Massachusetts}, \postcode{02453}, \country{USA}}. Email: yangyangwang@brandeis.edu}


\abstract{Homeostasis is widely observed in biological systems and refers to their ability to maintain an output quantity approximately constant despite variations in external disturbances. Mathematically, homeostasis can be formulated through an input–output function mapping an external parameter to an output variable. Infinitesimal homeostasis occurs at isolated points where the derivative of this input–output function vanishes, allowing tools from singularity theory and combinatorial matrix theory to characterize and classify homeostatic mechanisms in terms of network topology. Although the theoretical framework allows homeostasis subnetworks to be identified directly from combinatorial structures of the input–output network without numerical simulation, the required combinatorial enumeration becomes increasingly intractable as network size grows. Moreover, the reliance on advanced graph-theoretic concepts limits its broader accessibility and practical use across disciplines, particularly in biological applications.
To overcome these limitations, we develop a Python-based algorithm that automates the identification of homeostasis subnetworks and their associated homeostasis conditions directly from network topology. Given an input–output network specified solely by its connectivity structure and the designation of input and output nodes, the algorithm automatically identifies the relevant graph-theoretical structures and enumerates all homeostatic mechanisms. We demonstrate the applicability of the algorithm across a range of biological examples, including small and large networks, networks with a single input parameter (with single or multiple input nodes), multiple input parameters, and cases where input and output coincide. This wide applicability stems from our extension of the theoretical framework from single-input–single-output networks to networks with multiple input nodes through an augmented single-input-node representation.
The resulting computational framework provides a scalable and systematic approach to classifying homeostatic mechanisms in complex biological networks, facilitating the application of advanced mathematical theory to a broad range of biological systems.}


\maketitle

\section{Introduction}


Homeostasis is the ability to maintain internal stability while facing external perturbations.  
Usually there is an \emph{input-output function} $z(\II)$ that is required to be approximately constant over a range of external stimuli $\II$.
Many researchers have emphasized that homeostasis is an important phenomenon in biology.  
For example, the extensive work of Nijhout, Reed, Best and collaborators~\cite{NRBU04, RLN10, best2009homeostatic, NR14, nijhout2017systems,RBGSN17, NBR15, NBR18} considers networks whose nodes represent the concentrations of certain biochemical substrates.
Further examples include regulation of cell number and size \cite{L13}, the control of sleep~\cite{WRCD99}, and the expression level regulation in housekeeping genes~\cite{AGS18}.  
The literature is huge, and these articles are a small sample. 

In related work under the name of \emph{perfect adaptation} Ma \etal~\cite{MTELT09} identify, among all three-node enzyme networks with Michelis-Menten chemical kinetics, two networks that are able to \emph{adapt}, that is, to reset themselves after responding to an external stimulus.
This amounts to ask for the input-output function to be independent of the external stimuli $\II$, i.e., $z=$ constant.
{\em Perfect adaptation} has been used widely, for example, in ecology, chemistry, and control engineering (cf.~\cite{MTELT09, AM13, TM16, F16, QD18, AL18, DQMS18, ALGBSK19,hong2025topological,drengstig2012robust}).

Golubitsky and Stweart \cite{GS17} introduced the notion of \emph{infinitesimal homeostasis}~\cite{GS17} that is intermediary between perfect adaptation and general homeostasis.  
Consider system of ordinary differential equations (ODEs) depending on an \emph{input parameter} $\II$, which varies over a range of external stimuli.
Suppose there is a family of equilibrium points $X(\II)$ and an observable $\phi$, such that the \emph{input-output function} $z(\II)=\phi(X(\II))$ is well-defined on the range of $\II$.
In this situation, we say that the system exhibits \emph{homeostasis} if, under variation of the input parameter $\II$, the input-output function $z(\II)$ remains approximately constant over the interval of external stimuli.
We say that the system exhibits \emph{infinitesimal homeostasis} at $\II_0$ if $z'(\II_0) = 0$, where $'$ indicates differentiation with respect to $\II$.
Obviously, \emph{perfect adaptation} (or \emph{perfect homeostasis}) occurs when
$z'\equiv 0$.

Therefore, perfect homeostasis implies infinitesimal homeostasis at all points and infinitesimal homeostasis at $\II_0$ implies homeostasis in a neighborhood of $\II_0$.
It is not difficult to see that the converse do not hold in both cases.
One of the benefits of working with the notions of infinitesimal homeostasis is that it amounts to a very simple property to check, namely, if a derivative of a function is zero at some point.
On the other hand, this simple property opens up to the possibility to apply the powerful theorems of singularity theory to describe the topological structure of input-output functions.
Indeed, already in \cite{GS17} the authors are able to show that scalar input-output functions $z:\R\to\R$ can be classified by \emph{Elementary Catastrophe Theory} \cite{thom1969,thom1975,zeeman1977}.

A particularly fruitful approach to understand the mechanisms that lead to homeostasis is to regard homeostasis as a network concept.
This is motivated by the abundance of networks in biology, specially connected to ODE modeling, see e.g. \cite{NBR14,RBGSN17,F16}.
In fact, almost all ODE models in biology come together with a `wiring diagram', or a network, describing the interactions
among the elements in the model.
The recently published book \cite{gs2023} gives an exposition of a formal framework for studying networks of coupled ODEs, that have been developed by the authors and collaborators in the past decade. 
More precisely, a network of coupled ODEs is determined by directed graph whose nodes and edges are classified into types. 
Nodes, or cells, represent the variables of a component ODE. 
Edges, or arrows, between nodes represent couplings
from the tail node to the head node. 
Nodes of the same type have the same phase spaces (up to a canonical identification); edges of the same type represent identical couplings.

In the setting of networks of coupled ODEs there is a preferable class of observables, namely the output of the node variables (coordinate functions).
Network systems are distinguished from large systems by the ability to keep track of the output from each node individually. 
Now homeostasis can be naturally defined as the fact that the output of a node variable (the `output node') is held approximately constant as other variables (other nodes) vary (perhaps wildly) under variation of an input parameter that affects another node (the `input node').
Placing homeostasis in the general context of network dynamics leads naturally to the methods reviewed here.

A special kind of network of coupled ODE, with two distinguished nodes (one input node and one output node)
is called an \emph{input-output network} or \emph{single-input-single-output network}.
It was introduced in \cite{WHAG21} in their study of homeostasis on $3$-node networks.
Wang \etal~\cite{WHAG21} extended the notion of \emph{input-output network} to arbitrary large networks and developed a combinatorial theory for the classification of `homeostasis types' in such networks.
A \emph{homeostasis type} is essentially a `combinatorial mechanism' that causes homeostasis in the full network and is represented by specific subnetworks, called \emph{homeostasis subnetworks}.
The homeostasis subnetworks can be subdivided into two classes called \emph{structural} and \emph{appendage}.

The motivation for the term \emph{structural homeostasis} comes from \cite{RBGSN17}, where the authors identify the feedforward loop as one of the homeostatic motifs in $3$-node biochemical networks.
In this case, hemostasis is caused by a `balance' along the directed paths of feedforward loop, forming the 'structural backbone' of the input-output network.
The intuition behind the term \emph{appendage homeostasis} is that homeostasis is generated by a cycle of nodes attached to nodes in `structural backbone' of the input-output network.
Hence, the structural and appendage classes are abstract generalizations of the usual `feedforward' and `feedback' mechanisms. 
A striking outcome of \cite{WHAG21} is that they do not specify any homeostasis generating mechanisms at the outset and they find \emph{a posteriori} that there are essentially only the two types of homeostasis generating mechanisms: \emph{generalized feedback} (structural) and \emph{generalized feedforward} (appendage).

Currently, homeostasis subnetworks are recognized by manually computing certain combinatorial structures of the input-output network (see section \ref{sec:homeostasis-theory} for the definitions). 
However, this process would pose significant computational challenges in the case of large networks. 
As the network size increases, identifying these combinatorial structures becomes increasingly untractable. 

To address this issue, in this paper we present a Python-based algorithm that automates the identification of homeostasis subnetworks of an input-output network. The main contributions of this work are as follows:
\begin{itemize}
    \item We extend the existing theory developed for single-input-single-output network \cite{WHAG21} to more general input-output networks with multiple inputs or multiple input nodes.
    \item We develop a fully automated algorithm that identifies all homeostasis subnetworks and derives the corresponding homeostasis conditions directly from network topology.  
    \item We demonstrate the effectiveness of the approach through biologically motivated examples, including networks ranging from small to moderately large scale, with single or multiple inputs, single or multiple input nodes, as well as networks in which the input coincides with the output node. 
\end{itemize}

The remainder of our paper is organized as follows. In Section \ref{sec:singleinput}, we give the basic definitions of infinitesimal homeostasis, input-output networks, admissible system, and the homeostasis matrix, together with a number of combinatorial terms required for the theoretical characterization of homeostasis in an input-output network. These combinatorial structures enable the identification and classification of all the homeostasis subnetwork motifs based only on the network architecture, without performing numerical simulations on model equations. In Section \ref{sec:multipleinput}, we extend the theory developed for the single-input framework \cite{WHAG21} to more general input-output networks with multiple inputs or multiple input nodes. In Section \ref{S:algorithm-overview}, we describe the implementation and usage of our Python-based classification algorithm based on the theory developed in Section \ref{sec:homeostasis-theory}. We illustrate the workflow using a 12-node example network whose homeostatic structure has been fully analyzed in \cite{WHAG21}. In Section \ref{sec:application}, we apply the algorithm to several biological network examples ranging from small to relatively large networks. Finally, we conclude in Section \ref{sec:discussion}.

\section{Infinitesimal Homeostasis in Input-Output Networks}
\label{sec:homeostasis-theory}

\subsection{Single-Input Single-Output Networks}
\label{sec:singleinput}

We start with the case of single-input single-output input-output networks $\GG$, see \cite{GW20,WHAG21}.

The types of homeostasis that admissible systems of differential equations of $\GG$ can exhibit are characterized by the topology of so called homeostasis subnetworks.
Given an input-output network $\GG$ there is an algorithm for determining the homoestasis subnetworks of $\GG$.
In order to present the algorithm, we briefly review the basic definitions and results of \cite{WHAG21}.

As noted previously \cite{GS17,RBGSN17,GW20}, a straightforward application of Cramer's rule gives a formula for determining  infinitesimal homeostasis points.  See Lemma~\ref{L:det}.

An input-output network $\GG$ is a directed graph consisting of $n+2$ nodes. 
There is an {\em input node} $\iota$, an {\em output node} $o$, and $n$ {\em regulatory nodes} $\rho = (\rho_1,\ldots,\rho_n)$.
An {\em admissible system} associated with the input-output network $\GG$ is a parameterized system of ODEs
\begin{equation}\label{eq:ad io}
\dot{X} = F(X, \II)
\end{equation}
where $X = (x_\iota,x_\rho,x_o) \in \R^{n+2}$ are the node state variables, $\II\in\R$ is the {\em external input parameter}, and $F = (f_\iota,f_{\rho},f_o)$ is the associated vector field.
Explicitly, \eqref{eq:ad io} is the system 
\begin{equation} \label{eq: io}
\begin{split}
\dot{x}_\iota & = 
f_\iota(x_\iota,x_\rho, x_o,\II) \\
\dot{x}_\rho & =
f_\rho(x_\iota,x_\rho,x_o) \\
\dot{x}_o & = f_o(x_\iota,x_\rho,x_o)
\end{split}
\end{equation}
The compatibility of $F$ with the network $\GG$ is given by the following conditions:
\begin{enumerate}[label=(\alph*)]
\item $f_j$ depends on node $\ell$ only if there is an arrow in the network $\GG$ from $\ell\to j$. 
\item $f_\iota$ is the only vector field component that depends explicitly on $\II$ and $f_{\iota,\II}\neq 0$ generically.
\end{enumerate}
We write $f_{i,x_j}$ to denote the partial derivative of $f_i$ with respect to $j$ at $(X_0,\II_0)$.

\begin{remark} \rm \label{rmk:input=output}
In \cite{WHAG21} the authors explicitly exclude the possibility that the output node is one of the input nodes.
This assumption is included purely for the sake of convenience. 
In fact, all the definitions and results remain valid (including the algorithm for classification of homeostasis subnetworks) when the input node is the same as the output nodes (see \cite{antoneli2025a} and sub-section \ref{ssec:input_equal_output}).
In example \ref{ss:zinc_example} we apply our algorithm and discuss a network with input node equals the output node.
\END
\end{remark}

In order to define the notion of `infinitesimal homeostasis' in the context of input-output networks, assume that $X_0$ is a linearly stable equilibrium of \eqref{eq: io} at $\II=\II_0$.
Stability of $X_0$ implies that there is a unique stable equilibrium at $X(\II)=\big(x_\iota(\II),x_\rho(\II),x_o(\II)\big)$ as $\II$ varies on neighborhood of $\II_0$.

\begin{definition} \normalfont \label{def:inf_homeo}
The {\em input-output} function of system \eqref{eq: io}, at the family of equilibria $\big(X(\II),\II\big)$, is the function $\II\to x_o(\II)$, that is, the projection of $X(\II)$ onto the coordinate $x_o$.
We say that the input-output function $x_o(\II)$ exhibits \emph{infinitesimal homeostasis} at $\mathcal{I}_0$, if
\begin{equation} \label{def:inf_homeo_io}
x_o^{\prime}(\mathcal{I}_0) = 0
\end{equation}
where $'$ indicates differentiation with respect to $\II$.
We say that the input-output function $x_o(\II)$ exhibits \emph{perfect homeostasis} or \emph{perfect adaptation} if
$x_o^\prime \equiv 0$, that is, $x_o(\II)$ is constant.
\END
\end{definition}

We use the following notation.  
Let $J$ be the $(n+2)\times(n+2)$ Jacobian matrix of \eqref{eq: io} and let $\PH$ be the $(n+1)\times(n+1)$ \emph{homeostasis matrix} given by dropping the first row and the last column of $J$:
\begin{equation} \label{xo'_reduced2}
J = \Matrixc{ f_{\iota, x_\iota}   &  f_{\iota, x_\rho} & f_{\iota, x_o} \\
  f_{\rho, x_\iota}   &  f_{\rho, x_\rho} & f_{\rho, x_o} \\
  f_{o, x_\iota} &  f_{o, x_\rho} & f_{o, x_o} }
\quad 
H = \Matrixc{ f_{\rho, x_\iota}   &  f_{\rho, x_\rho}\\ f_{o, x_\iota} &  f_{o, x_\rho}}
\end{equation}
Here all partial derivatives $f_{\ell,x_j}$ are evaluated at the equilibrium $X_o$.

\begin{lemma} \label{L:det}
Let $(X_0,\II_0)$ be an asymptotically stable equilibrium of \eqref{eq: io}.  The input-output function $x_o(\II)$ satisfies
\begin{equation} \label{xo'_new}
x_o' = \pm\frac{f_{\iota, \II}}{\det(J)} \det(\PH)
\end{equation}
Hence, $\II_0$ is a point of infinitesimal homeostasis if and only if 
\begin{equation} \label{xo'_reduced}
\det(\PH)= 0
\end{equation}
at $(X_o, \II_0)$.
\end{lemma}

Homeostasis in a given network $\GG$ can be determined by analyzing a simpler network that is obtained by eliminating certain nodes and arrows from $\GG$.  We call the network formed by the remaining nodes and arrows the {\em core subnetwork}.

\begin{definition} \rm \label{D:updown}
A node $\tau$ in a network $\GG$ is \emph{downstream} from a node $\rho$ in $\GG$ if there exists a path in $\GG$ from $\rho$ to $\tau$.  Node $\rho$ is \emph{upstream} from node $\tau$ if $\tau$ is downstream from $\rho$.
\END
\end{definition}

\begin{definition} \rm \label{D:core}
The input-output network $\GG$ is a {\em core network} if every node in $\GG$ is both upstream from the output node $o$ and downstream from the input node $\iota$. 
$\GG_c$ is the {\em core subnetwork} whose nodes are the nodes in $\GG$ that are both upstream from the output and downstream from the input and whose arrows are the arrows in $\GG$ whose head and tail nodes are both nodes in $\GG_c$.
\END
\end{definition}

Core networks always have that $\det(H) \neq 0$, as multivariate polynomial in $f_{\ell,x_j}$ (see \cite[App. B]{madeira2022homeostasis}).
In other words, in non-core networks $\det(H)$ might vanish identically.

The computation of infinitesimal homeostasis reduces to solving $\det(H)=0$, where $H$ is the homeostasis matrix associated with the core subnetwork of $\GG$. 
Applying Frobenius-K\"{o}nig theory \cite{BR91,S77} to $H$, \cite{WHAG21} shows that there is a unique factorization 
$\det(H) = \det(B_1)\cdots \det(B_m)$.
Each block submatrix $B_\eta$ is an irreducible component of $H$, with $\det(B_\eta)=0$ and $\det(B_\chi)\neq 0$ for all $\chi\neq \eta$ being a defining condition for infinitesimal homeostasis.
This is generic when there is only one input parameter.
Furthermore, one can associate a homeostasis subnetwork $\KK_\eta$ of $\GG$ with each $B_\eta$.

\begin{definition} \rm \label{D:simple_complementary}
Let $\GG$ be a core input-output network.
\begin{enumerate}[label=(\alph*)]
\item A directed path connecting nodes $\rho$ and $\tau$ is called a \emph{simple path} if it visits each node on the path at most once.  
\item An {\em $\iota o$-simple path} is a simple path connecting the input node $\iota$ to the output node $o$.
\item A node in $\GG$ is {\em simple} if the node lies on an $\iota o$-simple path and {\em appendage} if the node is not simple.
\item A {\em super-simple} node is a simple node that lies on every $\iota o$-simple path.
\end{enumerate}
Nodes $\iota$ and $o$ are super-simple since by definition these nodes are on every $\iota o$-simple path.
\END
\end{definition}

Observe that super-simple nodes are well ordered (by downstream ordering) and hence adjacent super-simple pairs of nodes can be identified. 

\paragraph{Appendage Homeostasis}
\label{ss:prop-appendage}

If $B_\eta$ corresponds to an irreducible appendage block, then the associated homeostasis subnetwork $\KK_\eta$ consists of only appendage nodes and its Jacobian matrix is given by $B_\eta$. 
Combinatorial characterization of appendage homeostasis networks requires the following definitions.

\begin{definition} \rm \label{D:appendage_terms} 
Let $\GG$ be a core input-output network.
\begin{enumerate}[label=(\alph*)]
\item The {\em appendage subnetwork} $\AA_\GG$ of $\GG$ is the subnetwork consisting of all appendage nodes and all arrows in $\GG$ connecting appendage nodes. 
\item The {\em complementary subnetwork} of an $\iota o$-simple path $S$ is the subnetwork $\CC_S$ consisting of all nodes not on $S$ and all arrows in $\GG$ connecting those nodes.
\item Nodes $\rho_i,\rho_j$ in $\AA_\GG$ are {\em path equivalent} if there exists paths in $\AA_\GG$ from $\rho_i$ to $\rho_j$ and from $\rho_j$ to $\rho_i$.  An {\em appendage path component} is a path equivalence class in $\AA_\GG$.
\item A {\em cycle} is a path whose first and last nodes are identical.
\item Let $\KK\subset \AA_\GG$ be a subnetwork.  We say that $\KK$ satisfies the {\em no cycle condition} if for every $\iota o$-simple path $S$, nodes in $\KK$ do not form a cycle with nodes in $\CC_S\setminus\KK$. \END
\end{enumerate}
\end{definition}


\begin{remark} \rm \label{R:AiBi}
Nodes in the appendage subnetwork $\AA_\GG$ can be written uniquely as the disjoint union 
\begin{equation} \label{e:AiBi}
\AA_\GG = (\AA_1\dot{\cup}\cdots\dot{\cup}\AA_s)\; \dot{\cup}\;(\BB_1\dot{\cup}\cdots\dot{\cup}\BB_t)
\end{equation}
where each $\AA_i$ is an appendage path component that satisfies the no cycle condition and each $\BB_i$ is an appendage path component that violates the no cycle condition. Moreover, each $\AA_i$ (resp. $\BB_i$) can be viewed as a subnetwork of $\AA_\GG$ by including the arrows in $\AA_\GG$ that connect nodes in $\AA_i$ (resp. $\BB_i$).  We call $\AA_i$ a {\em no cycle appendage path component} and $\BB_i$ a {\em cycle appendage path component}.
\END
\end{remark}

\paragraph{Structural Homeostasis}
\label{ss:prop-structural}

If $B_\eta$ corresponds to an irreducible structural block, then $\KK_\eta$ has two adjacent super-simple nodes and these super-simple nodes are the input node $\ell$ and the output node $j$ in $\KK_\eta$.  
In addition, the network $\KK_\eta$ contains no backward arrows. 
That is, no arrows of $\KK_\eta$ go into the input node $\ell$ nor out of the output node $j$.

\begin{definition} \rm \label{D:structure_net}
The {\em structural subnetwork} $\sS_\GG$ of $\GG$ is the subnetwork whose nodes are either simple or in a cycle appendage path component $\BB_i$ (see Remark~\ref{R:AiBi}) and whose arrows are arrows in $\GG$ that connect nodes in $\sS_\GG$.
\END
\end{definition}

All structural homeostasis subnetworks are contained in $\sS_\GG$, which is an input-output network. That is,  $\GG$ and $\sS_\GG$ have the same simple, super-simple, input, and output  nodes.
Moreover, every non-super-simple simple node lies between two adjacent super-simple nodes.

\begin{definition}\rm \label{D:supsimpnet}
Let $\rho_1, \rho_2$ be adjacent super-simple nodes.
\begin{enumerate}[label=(\alph*)]
\item A simple node $\rho$ is {\em between} $\rho_1$ and $\rho_2$ if there exists an $\iota o$-simple path that includes $\rho_1$ to $\rho$  to $\rho_2$ in that order. 
\item The {\em super-simple subnetwork}, denoted $\LL(\rho_1,\rho_2)$, is the subnetwork whose nodes are simple nodes between $\rho_1$ and $\rho_2$ and whose arrows are arrows of $\GG$ connecting nodes in $\LL(\rho_1,\rho_2)$. \END
\end{enumerate}
\end{definition}

It follows that all $\LL(\rho_1,\rho_2)$ are contained in $\sS_\GG$ and each appendage node in $\sS_\GG$ connects to exactly one $\LL$. 
Hence, a super-simple subnetwork $\LL\subset \sS_\GG$ can be expanded to a super-simple structural subnetwork $\LL'\subset \sS_\GG$ as follows.

\begin{definition} \rm \label{def:L'}
Let $\rho_1$ and $\rho_2$ be adjacent super-simple nodes in $\GG$.  The \textit{super-simple structural subnetwork} $\LL'(\rho_1,\rho_2)$ is the input-output subnetwork consisting of nodes in $\LL(\rho_1,\rho_2)\cup \BB$ where $\BB$ consists of all appendage nodes that form cycles with nodes in $\LL(\rho_1,\rho_2)$; that is, all cycle appendage path components that connect to $\LL(\rho_1,\rho_2)$.   Arrows of $\LL'(\rho_1,\rho_2)$ are arrows of $\GG$ that connect nodes in $\LL'(\rho_1,\rho_2)$.  
Thus, $\rho_1$ is the input node and $\rho_2$ is the output node of $\LL'(\rho_1,\rho_2)$.
\END
\end{definition}

\subsubsection{Networks with Input = Output}
\label{ssec:input_equal_output}

Antoneli \etal~\cite{andrade2022,antoneli2025a,antoneli2026} provides a detailed discussion of the infinitesimal homeostasis formalism in the case where the input and the output nodes are the same. 
We briefly recall the main points in the following.

When the network $\mathcal{G}$ has the same node as the input and output nodes the corresponding variables coincide $x_\iota=x_o$ and we call the network an \emph{input=output network} and $\iota$ the input $=$ output node.

In this case the vector of state variable is $X=(x_{\iota},x_{\rho})\in\mathbb{R}\times\mathbb{R}^N$ and the system of ODE's \eqref{eq:ad io} becomes
\begin{equation}  \label{admissible_systems_ODE_IO}
\begin{aligned}
\dot{x}_{\iota} & = f_{\iota}(x_{\iota}, x_{\rho}, \mathcal{I}) \\
\dot{x}_{\rho} & = f_{\rho}(x_{\iota}, x_{\rho})\\
\end{aligned}
\end{equation}
Let $J$ be the $(N+1)\times (N+1)$ Jacobian matrix of an admissible vector field $F=(f_{\iota},f_{\sigma})$, that is,
\begin{equation} \label{jacobian_general}
J = \begin{pmatrix}
  f_{\iota, x_{\iota}}   &  f_{\iota, x_\rho} \\
  f_{\rho, x_{\iota}}   &  f_{\rho, x_\rho} 
\end{pmatrix}
\end{equation}
Now, the $N\times N$ \emph{homeostasis matrix} $H$ is obtained from $J$ by removing the \emph{first row} and the \emph{first column}:
\begin{equation}
\label{homeostasis_matrix_definition}
H = 
\begin{pmatrix}
f_{\rho, x_\rho}
\end{pmatrix}
\end{equation}
In both \eqref{jacobian_general} and \eqref{homeostasis_matrix_definition} partial derivatives $f_{\ell,x_j}$ are evaluated at the equilibrium $\big(X(\mathcal{I}),\mathcal{I}\big)$.

The main difference between the homeostasis matrix \eqref{homeostasis_matrix_definition} and the homeostasis matrix of network with distinct input and output nodes is that the former contains only the partial derivatives associated with the regulatory nodes, while the latter contains partial derivatives involving the input and the output nodes, as well.
In fact, the matrix $H$ in eq. \eqref{homeostasis_matrix_definition} is the Jacobian matrix of the subnetwork generated by the regulatory nodes.

As we have seen, the classification of homeostasis subnetworks stats with the reduction the the core subnetwork. 
The same definition of core network and core-equivalence apply to an input $=$ output network.
However, in the input $=$ output case, the condition that a node $\rho$ that is both upstream from the output node and downstream from the input node takes a special form.

\begin{lemma}
Let $\mathcal{G}$ be an input $=$ output network, with input $=$ output node $\iota$.
A regulatory node $\rho$ belongs to the core subnetwork $\mathcal{G}_c$ if and only if $\rho$ belongs to a cycle that contains the input $=$ output node $\iota$.
\end{lemma}
\begin{proof}
A directed path starting and ending at the same node is a cycle.
\end{proof}

\begin{corollary}
An input $=$ output network $\mathcal{G}$ is a core network if and only if every regulatory node of $\mathcal{G}$ belongs to a cycle that contains the input $=$ output node.
\end{corollary}

In a core network $\mathcal{G}$ where the input node is distinct from the output node, a simple node is always downstream from the input node and upstream from the output node, but not `the other way around'. 
That is, if a node is \emph{downstream} the \emph{output node} and/or \emph{upstream} the \emph{input node} then it must be an \emph{appendage node}. 
Indeed, if $\mathcal{G}$ is a core network then every node is downstream from the input node and upstream from the output node. 
Then a node that satisfies `the other way around' condition above must be on an $\iota o$-path that cycles around the input and/or the output node. 
We note the not every appendage node is of this type, that is, it may be on an $\iota o$-path that cycles around other regulatory nodes.
Hence, we have the following.

\begin{lemma} \label{THM:APPENDAGE}
Let $\mathcal{G}$ be a core input $=$ output network.
Then every node of $\mathcal{G}$ is appendage.
\end{lemma}
\begin{proof}
Since every node in a core input $=$ output network forms a cycle with the input $=$ output node it follows that every node is \emph{downstream} the \emph{output node} and \emph{upstream} the \emph{input node}.
Thus, every node is appendage.
\end{proof}

\begin{corollary}
Structural homeostasis does not exist in input $=$ output networks.
\end{corollary}

The next theorem summarizes the classification of homeostasis subnetworks of an input $=$ output network.

\begin{theorem} \label{THM:APPENDAGE_NET}
If $\mathcal{G}$ is an input $=$ output network then the appendage subnetwork of $\mathcal{G}$ is exactly the subnetwork of $\mathcal{G}$ generated by all regulatory nodes.
Moreover, the irreducible factors of $\det(H)$ correspond to the appendage path components $\mathcal{A}_k$ and are given by $\det(J({\mathcal{A}_k}))$. 
Here $J({\mathcal{A}_k})$ is the jacobian matrix associated with the network $\mathcal{A}_k$.
\end{theorem}
\begin{proof}
By lemma \ref{THM:APPENDAGE} it follows that the set of nodes of $\mathcal{A}_\mathcal{G}$ is exactly the set of regulatory nodes. Hence they generate the same network.
The, it follows from \cite[Thm. 7.1]{WHAG21}, that all irreducible blocks of the homeostasis matrix correspond to the appendage path components. Moreover, since there are no simple nodes the ``no cycle condition'' is trivially satisfied by all appendage path components.
Finally, it follows from \cite[Thm. 5.4]{WHAG21} that the corresponding irreducible factor of $\det(H)$ is of the form $\det(J({\mathcal{A}_k}))$.
\end{proof}

\subsubsection{Algorithm for the Classification of Homeostasis Subnetworks} 
\label{sec:algorithm}

Here is the algorithm for enumerating subnetworks corresponding to the $m$ homeostasis blocks (see \cite{WHAG21}).

\paragraph{Step 1:} Determining the appendage homeostasis subnetworks from $\AA_\GG$.  
Let 
\begin{equation} \label{e:AHB}
\AA_1\;,\ldots,\;\AA_s
\end{equation}
be the no cycle appendage path components of $\AA_\GG$ (see Remark~\ref{R:AiBi}). 
These appendage path components are the subnetworks $\KK_\eta$ that correspond to appendage homeostasis blocks. 
In addition, there are $s$ independent defining conditions for appendage homeostasis given by the determinants of the Jacobian matrices $\det(J_{\AA_i}) = 0$ for $i = 1,\ldots, s$.

\paragraph{Step 2:} Determining the structural homeostasis subnetworks from $\sS_\GG$. 
Let $\iota = \rho_1 > \rho_2 > \ldots > \rho_{q+1} = o$ be the super-simple nodes in $\sS_\GG$ in downstream order.
Then, up to core equivalence, the $q$ super-simple structural subnetworks 
\begin{equation} \label{e:SHB}
\LL'(\iota,\rho_2),\; \LL'(\rho_2,\rho_3) \;,\ldots,\; \LL'(\rho_{q-1},\rho_q),\; \LL'(\rho_q,o)
\end{equation}
are the subnetworks $\KK_\eta$ that correspond to structural homeostasis blocks. 
In addition, there are  $q$ defining conditions for structural homeostasis blocks given by the determinants of the homeostasis matrices of the input-output networks: $\det\big(H(\LL'(\rho_i,\rho_{i+1}))\big) = 0$ for $i = 1,\ldots, q$.

Therefore, the $m = s + q$ subnetworks listed in \eqref{e:AHB} and \eqref{e:SHB} enumerate the appendage and structural homeostasis subnetworks in $\GG$.

\subsection{Generalization to Multiple-Input Single-Output Networks}
\label{sec:multipleinput}

Here we explain how to extend the theory developed for single-input framework (algorithm \ref{sec:algorithm}) to more general input-output networks, by reducing to the single-input single-output case.

The first step is to consider 
single-input single-output input-output network with multiple input nodes, but a single input parameter.
Then we can deal with the general case of multiple-input single-output input-output networks, that is, input-output network with multiple input nodes and multiple input parameters.

\subsubsection{Single-Input Single-Output Networks with Multiple Input Nodes}
\label{sec:siso_mi}

In \cite{madeira2022homeostasis}, the authors extend the approach of \cite{WHAG21} to deal with the multiple input node case.
They consider the influence of the input parameter on each input node separately by constructing a distinguished set of subnetworks with a single input node.
This strategy effectively reduces the problem to the single input node case.
Then they show how to put together the classification of homeostasis subnteworks of the the networks in the distinguished set of subnetworks with a single input node.

Here, we will present an alternative approach to the classification of the homeostasis subnetworks of networks with multiple input nodes that allows us to directly apply algorithm \ref{sec:algorithm}.
Let us briefly recall some definitions and results from \cite{madeira2022homeostasis} before introducing the alternative approach.

A \emph{multiple input-node input-output network} is a network $\mathcal{G}$ with $n$ distinguished \emph{input nodes} $\iota=\{\iota_{1}, \ldots, \iota_{n}\}$, all of them associated to the same input parameter $\mathcal{I}$, one distinguished \emph{output node} $o$, and $N$ \emph{regulatory nodes} $\rho=\{\rho_1,\ldots,\rho_N\}$.
The associated network systems of differential equations have the form
\begin{equation} \label{admissible_systems_ODE_multiple_input_nodes}
\begin{aligned}
\dot{x}_{\iota} & = f_{\iota}(x_{\iota}, x_{\rho}, x_{o}, \mathcal{I}) \\
\dot{x}_{\rho} & = f_{\rho}(x_{\iota}, x_{\rho}, x_{o})\\
\dot{x}_{o} & = f_{o}(x_{\iota}, x_{\rho}, x_{o})
\end{aligned}
\end{equation}
where $\mathcal{I}\in\mathbb{R}$ is an \emph{external input parameter} and $X=(x_{\iota},x_{\rho},x_o)\in\mathbb{R}^n\times\mathbb{R}^N\times\mathbb{R}$ is the vector of state variables associated to the network nodes.
Finally, assume that $\GG$ satisfies the following condition: the output node is downstream from all input nodes.

We write a vector field associated with the system \eqref{admissible_systems_ODE_multiple_input_nodes} as
\[
F(X,\mathcal{I})=(f_{\iota}(X,\mathcal{I}),f_\rho(X),f_o(X))
\]
and call it an \emph{admissible vector field} for the network $\mathcal{G}$.

Let $f_{j,x_\ell}$ denote the partial derivative of the $j^{th}$ node function $f_j$ with respect to the $\ell^{th}$ node variable $x_\ell$. We make the following assumptions about the vector field $F$ throughout:
\begin{enumerate}[label=(\alph*)]
\item The vector field $F$ is smooth and has an asymptotically stable equilibrium at $(X^*,\mathcal{I}^*)$. Therefore, by the implicit function theorem, there is a function $\tilde{X}(\mathcal{I})$ defined in a neighborhood of $\mathcal{I}^*$ such that $\tilde{X}(\mathcal{I}^*) = X^*$ and $F(\tilde{X}(\mathcal{I}), \mathcal{I}) \equiv 0$. 
\item The partial derivative $f_{j,x_\ell}$ can be non-zero only if the network $\mathcal{G}$ has an arrow $\ell\to j$, otherwise $f_{j,x_\ell} \equiv 0$.
\item Only the input node coordinate functions $f_{\iota_m}$ depend on the external input parameter $\mathcal{I}$ and the partial derivative of $f_{\iota_m,\mathcal{I}}$ generically satisfies
\begin{equation} \label{e:f_iota_I}
 f_{\iota_m,\mathcal{I}} \neq 0.
\end{equation}
\end{enumerate}

Let $J$ be the $(n+N+1)\times (n+N+1)$ Jacobian matrix of an admissible vector field $F=(f_{\iota},f_{\sigma},f_{o})$, that is,
\begin{equation} \label{jacobian}
J = \begin{bmatrix}
  f_{\iota, x_\iota}   &  f_{\iota, x_\rho} & f_{\iota, x_o} \\
  f_{\rho, x_\iota}   &  f_{\rho, x_\rho} & f_{\rho, x_o} \\
  f_{o, x_\iota} &  f_{o, x_\rho} & f_{o, x_o} 
\end{bmatrix}
\end{equation}
The $(n+N+1)\times (n+N+1)$ matrix $\langle H \rangle$ obtained from $J$ by replacing the last column by $(-f_{\iota,\mathcal{I}},0,0)^t$, is called \emph{generalized homeostasis matrix}:
\begin{equation} \label{weighted_homeostasis_matrix_definition}
\langle H \rangle = 
\begin{bmatrix}
f_{\iota, x_\iota} &  f_{\iota, x_\rho} & -f_{\iota, \mathcal{I}} \\
f_{\rho, x_\iota}&  f_{\rho, x_\rho} & 0 \\
f_{o, x_\iota} &  f_{o, x_\rho} & 0
\end{bmatrix} = \begin{bmatrix} f_{\iota_{1}, x_{\iota_{1}}} & \cdots & f_{\iota_{1}, x_{\iota_{n}}} & f_{\iota_{1}, x_{\rho}} & - f_{\iota_{1}, \mathcal{I}} \\
    \vdots & \ddots & \vdots & \vdots & \vdots \\
    f_{\iota_{n}, x_{\iota_{1}}} & \cdots & f_{\iota_{n}, x_{\iota_{n}}} & f_{\iota_{n}, x_{\rho}} & - f_{\iota_{n}, \mathcal{I}} \\
    f_{\rho, x_{\iota_{1}}} & \cdots & f_{\rho, x_{\iota_{n}}} & f_{\rho, x_{\rho}} & 0 \\
    f_{o, x_{\iota_{1}}} & \cdots & f_{o, x_{\iota_{n}}} & f_{o, x_{\rho}} & 0
    \end{bmatrix}
\end{equation}
Here all partial derivatives $f_{\ell,x_j}$ are evaluated at $\big(\tilde{X}(\mathcal{I}),\mathcal{I}\big)$.

\begin{lemma} \label{cramer_rule}
The input-output function $x_o(\mathcal{I})$ satisfies
\begin{equation} \label{xo'_MI}
x_o'(\mathcal{I}) = \frac{\det\!\big(\langle H \rangle\big) }{\det(J)}
\end{equation}
Here $\det(J)$ and $\det\!\big(\langle H \rangle\big)$ are evaluated at $\big(\tilde{X}(\mathcal{I}),\mathcal{I}\big)$. Hence, $\mathcal{I}_0$ is a point of infinitesimal homeostasis if and only if
\begin{equation} \label{xo'_reduced_MI}
\det\!\big(\langle H \rangle\big) = 0
\end{equation}
at the equilibrium $\big(\tilde{X}(\mathcal{I}_0),\mathcal{I}_0\big)$.
\end{lemma}

By expanding $\det(\langle H \rangle)$ with respect to the last column and each $\iota_k$ (input) row one obtains 
\begin{equation} \label{xo'_reduced_expand}
\det\!\big(\langle H \rangle\big) = \sum_{m=1}^n \pm f_{\iota_m,\mathcal{I}} \det(H_{\iota_m})
\end{equation}
The \emph{partial homeostasis matrix} $H_{\iota_m}$ is obtained from the Jacobian matrix $J$ of $F$ by dropping the last column and the $\iota_m$ row; see \eqref{definition_parcels_homeostasis_matrix} below.
\begin{equation} \label{definition_parcels_homeostasis_matrix}
    H_{\iota_{m}} = \begin{bmatrix} f_{\iota_{1}, x_{\iota_{1}}} & \cdots & f_{\iota_{1}, x_{\iota_{n}}} & f_{\iota_{1}, x_{\rho}}\\
    \vdots & \ddots & \vdots & \vdots \\ f_{\iota_{m-1}, x_{\iota_{1}}} & \cdots & f_{\iota_{m-1}, x_{\iota_{n}}} & f_{\iota_{m-1}, x_{\rho}} \\ f_{\iota_{m+1}, x_{\iota_{1}}} & \cdots & f_{\iota_{m+1}, x_{\iota_{n}}} & f_{\iota_{m+1}, x_{\rho}} \\  \vdots & \ddots & \vdots & \vdots \\
    f_{\iota_{n}, x_{\iota_{1}}} & \cdots & f_{\iota_{n}, x_{\iota_{n}}} & f_{\iota_{n}, x_{\rho}} \\
    f_{\rho, x_{\iota_{1}}} & \cdots & f_{\rho, x_{\iota_{n}}} & f_{\rho, x_{\rho}} \\
    f_{o, x_{\iota_{1}}} & \cdots & f_{o, x_{\iota_{n}}} & f_{o, x_{\rho}} \end{bmatrix}
\end{equation}

Note that when there is a single input node, i.e. $n=1$, Lemma \ref{cramer_rule} gives the corresponding result obtained in~\cite{WHAG21}. 
In this case, there is only one matrix $H_{\iota_m}=H$, which reduces to the \emph{homeostasis matrix} from single input node case.

The generalization of the notion of core network to the multiple input nodes is straightforward.

\begin{definition} \label{defining_core_networks} \rm
Let $\mathcal{G}$ be a network with input nodes $\iota_{1}, \ldots, \iota_{n}$ and output node $o$. 
We call $\mathcal{G}$ a \emph{core network} if every node in $\mathcal{G}$ is upstream from $o$ and downstream from at least one input node.
$\GG_c$ is the \emph{core subnetwork} whose nodes are the nodes in $\GG$ that are both upstream from $o$ and downstream from at least one input node and whose arrows are the arrows in $\GG$ whose head and tail nodes are both nodes in $\GG_c$.
\END
\end{definition}

In \cite{madeira2022homeostasis} the authors proceed to extend all the combinatorial properties of the single input node case to the multiple input node case, culminating with the definitions of homeostasis subnetworks.
Then the classification proceeds by carefully analyzing how these notions relate with the corresponding notions on each single input node subnetwork of the distinguished set.

Now, we are in position to describe our new approach to the classification of the homeostasis subnetwork of networks with multiple input nodes (see Theorem \ref{thm:GenH} below). 

Specifically, we augment the network $\GG$ by introducing a new input node $x_{I}$ that connects to all original input nodes and whose dynamics enforces $x_{I}(t) = \mathcal{I}$ asymptotically. 
This construction effectively converts an $n+N$-node multiple input node network to an $n+N+1$-node single input node network, allowing the theory and classification results we developed in \cite{WHAG21} for single input node networks to be applied directly. \\

\begin{theorem} \label{thm:GenH}
Let $\mathcal{G}$ be a single-input input-output network with $n$ distinguished input nodes 
$\iota = (\iota_1, \dots, \iota_n)$ receiving the same input from $\mathcal{I}\in \mathbb{R}$, $N$ regulatory nodes $\rho$, and one distinguished output node $o$.
Construct the \emph{augmented network} $\mathcal{G}^\diamond$ by introducing a new input node $I$ and converting the original input nodes $\iota = (\iota_1, \dots, \iota_n)$ into regulatory nodes, by adding directed edges 
$I \to \iota_j$ for $j = 1,\dots,n$. 
Extend the vector field $F=(f_{\iota},f_{\rho},f_o)$ to $F^\diamond=(f_{I},f_{\iota},f_{\rho},f_o)$ in such a way that the equation for the new input node variable  $\dot{x}_I(t)=f_I(x_{I},\mathcal{I})$ enforces $x_I(t)\equiv \II$ asymptotically. 
Then the infinitesimal homeostasis matrix $H^\diamond$ of $\mathcal{G}^\diamond$ is identical to the generalized homeostasis matrix $\langle H \rangle$ of $\mathcal{G}$ given by \eqref{weighted_homeostasis_matrix_definition}, up to a sign difference.
\end{theorem}

\begin{proof} 
Let $H^\diamond$ be the homeostasis matrix for the augmented network $\mathcal{G}^\diamond$. 
We prove below that $\det(H^\diamond) = \det(\langle H \rangle)$, as polynomials in $(f_{\ell,x_j})$, up to a sign.
Consequently, both formulations yield identical infinitesimal homeostasis conditions.
We define the dynamics of the new input node as 
\[
\dot{x}_{I}=f_I(x_{I},\mathcal{I}) := \mathcal{I}-x_I,
\]
which ensures $x_I(t) = \mathcal{I}$ asymptotically.
Moreover, because of the skew product from of $F^\diamond$, if the original system $F$ is stable at an equilibrium $(X_0,\II_0)$ then the extended system $F^\diamond$ is stable at the corresponding equilibrium $(\II_0,X_0,\II_0)$. 
Other choices of dynamics are possible. 
This is one working example. 
We then rewrite dynamics for the original set of input nodes as 
\[
\dot{x}_{\iota} = f_{\iota}(x_\iota,x_\rho,x_o, x_I),
\]
which now become regulatory nodes that do not directly depend on the input parameter $\mathcal{I}$. 
The Jacobian matrix $J^\diamond$ of $\mathcal{G}^\diamond$ becomes:
\[
J^\diamond = \begin{bmatrix}
f_{I, x_{I}} & f_{I, x_{\iota}} & f_{I, x_{\rho}} & f_{I, x_o} \\
f_{\iota, x_{I}} & f_{\iota, x_{\iota}} & f_{\iota, x_{\rho}} & f_{\iota, x_o} \\
f_{\rho, x_{I}} & f_{\rho, x_{\iota}} & f_{\rho, x_{\rho}} & f_{\rho, x_o} \\
f_{o, x_{I}} & f_{o, x_{\iota}} & f_{o, x_{\rho}} & f_{o, x_o}
\end{bmatrix}
\quad=\quad
\begin{bmatrix}
f_{I, x_{I}} & 0 & 0 & 0 \\
f_{\iota, x_{I}} & f_{\iota, x_{\iota}} & f_{\iota, x_{\rho}} & f_{\iota, x_o} \\
0 & f_{\rho, x_{\iota}} & f_{\rho, x_{\rho}} & f_{\rho, x_o} \\
0 & f_{o, x_{\iota}} & f_{o, x_{\rho}} & f_{o, x_o}
\end{bmatrix}
\]
By Lemma \ref{L:det}, the input-output function $x_o(\mathcal{I})$ of this augmented network satisfies 
\[
x_o'=\pm \frac{f_{I,\mathcal{I}}}{\det(J^\diamond)}\det(H^\diamond) = \pm \frac{\det(H^\diamond)}{\det(J)}
\]
where $H^\diamond$ is obtained from removing the first row and the last column of $J^\diamond$:
\begin{equation}
\begin{aligned}
H^\diamond &= \begin{bmatrix}
f_{\iota, x_{I}} & f_{\iota, x_{\iota}} & f_{\iota, x_{\rho}}\\
0 & f_{\rho, x_{\iota}} & f_{\rho, x_{\rho}}\\
0 & f_{o, x_{\iota}} & f_{o, x_{\rho}}
\end{bmatrix} 
\quad = \quad
\begin{bmatrix}
f_{\iota, \mathcal{I}} & f_{\iota, x_{\iota}} & f_{\iota, x_{\rho}}\\
0 & f_{\rho, x_{\iota}} & f_{\rho, x_{\rho}}\\
0 & f_{o, x_{\iota}} & f_{o, x_{\rho}}
\end{bmatrix}
\end{aligned}
\end{equation}
The second equality follows that, at the stable equilibrium, $x_I=\mathcal{I}$. 
Through two column exchanges, we obtain that $\det(H^\diamond) = \det(\langle H \rangle)$. 
\end{proof}

Since the multivariate polynomial 
$\det(H^\diamond)$ have the same irreducible factors of $\det(\langle H \rangle)$, we can apply algorithm \ref{sec:algorithm} to the augmented network $\GG^\diamond$.
Moreover, it follows from \cite[Thm. 3.26]{madeira2022homeostasis} that there exist exactly one irreducible factor of $\det(\langle H \rangle)$ that contains all terms of the form $f_{\iota_i,\II}$.
This irreducible factor is associated with the \emph{input counterweight} homeostasis type.
In the irreducible decomposition of $\det(H^\diamond)$ the input counterweight homeostasis type corresponds to one of the structural homeostasis types, the one associated with the unique irreducible factor that contains all terms of the form $f_{\iota_i,x_I}$ (since these terms correspond to $f_{\iota_i,\II}$).

It may happen that the output node coincides with one of the input nodes $\iota_j$.
Then, in principle, one would have to take into account the observations from sub-section \ref{ssec:input_equal_output}.
However, it is easy to see that the augmented network always has the new input node distinct from the output node and hence one is back to the usual case where the input and the output nodes are distinct \cite{WHAG21}. 

\subsubsection{Multiple-Input Single-Output Networks}\label{sec:miso}

In \cite{madeira2024}, the authors extend the approach of \cite{WHAG21, madeira2022homeostasis} to deal with the general multiple-input single output case.
Here, we will show that algorithm \ref{sec:algorithm} can be applied multiple times to certain single-input single-output networks and hence obtain a classification of homeostasis types in this case, as well.

Let us briefly recall some
definitions and results from  \cite{madeira2024} before explaining the strategy to apply algorithm \ref{sec:algorithm} to obtain a classification of homeostasis types.

A \emph{multiple-input single-output} network, or simply a \emph{multiple inputs network}, is a network $\mathcal{G}$ with $n$ distinguished \emph{input nodes} $\iota=\{\iota_{1}, \iota_{2}, \ldots, \iota_{n}\}$, all of them associated to at least one input parameter $\mathcal{I}_{M}$, $M = 1, \ldots, N$, one distinguished \emph{output node} $o$, and $r$ \emph{regulatory nodes} $\rho=\{\rho_1,\ldots,\rho_r\}$.
The associated system of differential equations has the form
\begin{equation} \label{admissible_systems_ODE_multiple_inputs}
\begin{aligned}
\dot{x}_{\iota} & = f_{\iota}(x_{\iota}, x_{\rho}, x_{o}, \mathcal{I}) \\
\dot{x}_{\rho} & = f_{\rho}(x_{\iota}, x_{\rho}, x_{o})\\
\dot{x}_{o} & = f_{o}(x_{\iota}, x_{\rho}, x_{o})
\end{aligned}
\end{equation}
where $\mathcal{I} = (\mathcal{I}_{1}, \cdots, \mathcal{I}_{N}) \in\mathbb{R}^{N}$ is the vector of \emph{input parameters}, or simple the vector of \emph{inputs}, and $X=(x_{\iota},x_{\rho},x_o)\in\mathbb{R}^n\times\mathbb{R}^r \times\mathbb{R}$ is the vector of state variables associated to the network nodes.

We write a vector field associated with the system \eqref{admissible_systems_ODE_multiple_inputs} as
\[
F(X,\mathcal{I})=(f_{\iota}(X,\mathcal{I}),f_\rho(X),f_o(X))
\]
and call it an \emph{admissible vector field} for the network $\mathcal{G}$.

Let $f_{j,x_i}$ denote the partial derivative of the $j^{th}$ node function $f_j$ with respect to the $i^{th}$ node variable $x_i$.  We make the following assumptions about the vector field $F$ throughout:
\begin{enumerate}[label=(\alph*)]
\item The vector field $F$ is smooth and has a linearly stable equilibrium at $(X^*,\mathcal{I}^*)$.
Therefore, by the implicit function theorem, there is a function $\tilde{X}(\mathcal{I})$ defined in a neighborhood of $\mathcal{I}^*$ such that $\tilde{X}(\mathcal{I}^*) = X^*$ and $F(\tilde{X}(\mathcal{I}), \mathcal{I}) \equiv 0$. 
\item The partial derivative $f_{j,x_i}$ can be non-zero only if the network $\mathcal{G}$ has an arrow $i\to j$, otherwise $f_{j,x_i} \equiv 0$.
\item Only the input node coordinate functions $f_{\iota_k}$ depend on at least one of the components of the vector of input parameters $\mathcal{I}$ and the partial derivative of $f_{\iota_k,\mathcal{I}_{M}}$ generically satisfies
\begin{equation} \label{e:f_iota_II}
 \frac{\partial f_{\iota_{k}}}{\partial \mathcal{I}_{M}} = f_{\iota_k,\mathcal{I}_{M}} \neq 0.
\end{equation}
for some $M = 1, \ldots, N$.
\end{enumerate}

Now the input-output function $x_o$ is multivariate, that is $x_o:\mathbb{R}^N\to\mathbb{R}$.
Infinitesimal homeostasis in a multiple inputs network is given by the critical points of $x_o(\mathcal{I})$, namely, the zeros of the gradient vector
\begin{equation} \label{gradient_vector}
    \nabla x_{o} = \left(\frac{\partial x_{o}}{\partial \mathcal{I}_{1}}, \frac{\partial x_{o}}{\partial \mathcal{I}_{2}}, \cdots, \frac{\partial x_{o}}{\partial \mathcal{I}_{N}} \right)
\end{equation}

Let $J$ be the $(n+r+1)\times (n+r+1)$ Jacobian matrix of an admissible vector field $F=(f_{\iota},f_{\sigma},f_{o})$, that is,
\begin{equation} \label{jacobian2}
J = \begin{pmatrix}
  f_{\iota, x_\iota}   &  f_{\iota, x_\rho} & f_{\iota, x_o} \\
  f_{\rho, x_\iota}   &  f_{\rho, x_\rho} & f_{\rho, x_o} \\
  f_{o, x_\iota} &  f_{o, x_\rho} & f_{o, x_o} 
\end{pmatrix}
\end{equation}
For each $1 \leq M \leq N$, consider the $(n+r+1)\times (n+r+1)$ matrix $\langle H_M \rangle$ obtained from $J$ by replacing the last column by $(-f_{\iota,\mathcal{I}_M},0,0)^t$, which is called \emph{$\mathcal{I}_M$-generalized homeostasis matrix}:
\begin{equation} \label{weighted_homeostasis_matrix_definition_mi}
\langle H_M \rangle = 
\begin{pmatrix}
f_{\iota, x_\iota} &  f_{\iota, x_\rho} & -f_{\iota, \mathcal{I}_M} \\
f_{\rho, x_\iota}&  f_{\rho, x_\rho} & 0 \\
f_{o, x_\iota} &  f_{o, x_\rho} & 0
\end{pmatrix}
\end{equation}
Here all partial derivatives $f_{\ell,x_j}$ are evaluated at $\big(\tilde{X}(\mathcal{I}),\mathcal{I}\big)$.

\begin{lemma} \label{cramer_rule2}
Let $x_o(\mathcal{I})$ be input-output function of a multiple inputs network.
The partial derivative of $x_o(\mathcal{I})$ with respect to the $M$-th parameter $\mathcal{I}_M$ satisfies
\begin{equation} \label{xo'}
\frac{\partial x_o\;}{\partial \mathcal{I}_M} = \frac{\det\langle H_M \rangle}{\det(J)}
\end{equation}
Here $\det(J)$ and $\det\langle H_M \rangle$ are evaluated at the equilibrium point $\tilde{X}(\mathcal{I})$. Hence, 
\begin{equation} \label{gradient_simplified}
    \nabla x_{o} = \frac{1}{\det (J)}\left(\det \langle H_{1} \rangle, \det \langle H_{2} \rangle, \ldots, \det \langle H_{N} \rangle \right)
\end{equation}
Moreover, $\mathcal{I}^0$ is a point of infinitesimal homeostasis if and only if
\begin{equation} \label{xo'_reduced3}
\det\langle H_M \rangle = 0
\qquad \text{for all} \quad 1 \leq M \leq N
\end{equation}
as a function of $\mathcal{I}$ evaluated at $\mathcal{I}^0$.
\end{lemma}

\begin{remark} \normalfont \label{rmk:case_N_1}
An explicit expression for $\det\langle H_M \rangle$ can be obtained by expanding it with respect to the last column and the $\iota_m$-th row:
\begin{equation} \label{xo'_reduced_expand_mi}
\det\langle H_M \rangle = \sum_{m=1}^n \pm f_{\iota_m,\mathcal{I}_M} \det(H_{\iota_m})
\end{equation}
Here $H_{\iota_m}$ is obtained from $H$ by
removing the last column and the $\iota_m$-th row.
When there is a single input, i.e. $N=1$, the gradient $\nabla x_{o}$ reduces to ordinary derivative $x_o'$ and \eqref{gradient_simplified} gives the formula for $x_o'$ obtained in~\cite{madeira2022homeostasis}.
When there is a single input and a single input node, $N=n=1$, there is only one matrix $H_{\iota_m}=H$, called the \emph{homeostasis matrix} and
\eqref{gradient_simplified} gives the corresponding formula for $x_o'$ obtained in~\cite{WHAG21}.
\END
\end{remark}

\begin{definition} \normalfont
Let $\mathcal{G}$ be a multiple inputs network.
The \emph{core subnetwork} $\mathcal{G}_{c}$ of $\mathcal{G}$ is the subnetwork whose nodes are: (i) the input nodes $\iota_{1}, \ldots, \iota_{n}$, (ii) the regulatory nodes $\rho$ that are upstream from the output node and downstream of at least one input node, and (iii) the output node $o$. 
The arrows of $\mathcal{G}_{c}$ are the arrows of $\mathcal{G}$ connecting the nodes of $\mathcal{G}_{c}$.
\END
\end{definition}

From this point we will depart from \cite{madeira2024} and incorporate the alternative approach from the previous section \ref{sec:siso_mi}.
Let $\mathcal{G}$ be a core multiple inputs network with inputs $\mathcal{I}_{1}, \ldots, \mathcal{I}_{N}$.
The \emph{$\mathcal{I}_{M}$-specialized network} $\mathcal{G}_{\mathcal{I}_{M}}$ is defined as the (single input prameter) input-output network consisting of the network $\mathcal{G}$ with the input nodes being exactly the nodes affected by the input parameter $\II_M$.
In principle, the specialized network $\mathcal{G}_{\mathcal{I}_{M}}$ is a multiple input node network with single input parameter $\mathcal{I}_M$, as studied in \cite{madeira2022homeostasis} and considered in the previous section \ref{sec:siso_mi}.
Now, when the specialized network $\mathcal{G}_{\mathcal{I}_{M}}$  has more than one input node we can apply the construction of Theorem \ref{thm:GenH} and replace $\mathcal{G}_{\mathcal{I}_{M}}$ by its augmented version $\mathcal{G}^\diamond_{\mathcal{I}_{M}}$.
In this way all specialized networks are single-input single output networks.

However, it may happen that, although the multiple inputs network $\GG$ is core, when one drops all the other input parameters, the network $\mathcal{G}_{\mathcal{I}_{M}}$ (or $\mathcal{G}^\diamond_{\mathcal{I}_{M}}$) become non-core.
Therefore, one needs to consider the corresponding core networks 
$\mathcal{G}^{c}_{\mathcal{I}_{M}}$ (or $\mathcal{G}^{\diamond\, c}_{\mathcal{I}_{M}}$).
Moreover, some of these single-input single-output networks might have the input the same as the output node.
In any case, the algorithm can deal with all these networks.

Now, in order to put together the results of computation of the homeostasis subnetworks for each network $\mathcal{G}_{\mathcal{I}_{M}}$ (or $\mathcal{G}^\diamond_{\mathcal{I}_{M}}$) we need one more definition.

\begin{definition} \normalfont
The \emph{vector determinant} associated to an input-output network is the vector-valued function defined by
\begin{equation}
\widehat{h}=
    \big(\det \langle H_{\mathcal{I}_1} \rangle, \det \langle H_{\mathcal{I}_2} \rangle, \ldots, \det \langle H_{\mathcal{I}_N} \rangle \big),
\end{equation}
where $\det \langle H_{\mathcal{I}_{M}} \rangle$, $M=1,\ldots,N$, is the determinant of generalized homeostasis matrix of the network $\mathcal{G}_{\mathcal{I}_{M}}^c$ or its augmented version $\left(\mathcal{G}^\diamond_{\mathcal{I}_{M}}\right)^c$, according with the situation.
The vector-valued function $\widehat{h}$ can be considered as a (formal) polynomial mapping on the `variables' $f_{j,x_i}$ and $f_{j,I_M}$.
\END
\end{definition}

The following theorem ensures that $\widehat{h}$ contains all the information to classify the homeostasis subnetworks (see \cite{madeira2024}).

\vspace{-5mm}

\begin{theorem} \label{zero_set}
The irreducible factors of $\widehat{h}$ are exactly the irreducible factors of $\nabla x_o$ that can induce homeostasis.
\end{theorem}

The K\"onig-Frobenius theorem~\cite{S77,BC09} (see also~\cite{WHAG21,madeira2022homeostasis}) imply that the components of the polynomial mapping $\widehat{h}$ can be factorized as the product of the determinants of the irreducible diagonal blocks of each $\langle H_{\mathcal{I}_M} \rangle$ (defined up to row and column permutations).
An irreducible block $B$ of some $\langle H_{\mathcal{I}_M} \rangle$ is called a \emph{homeostasis block}.
We can further collect the factors that are common irreducible diagonal blocks of all matrices $\langle H_{\mathcal{I}_M} \rangle$ and bring them to the front as scalar factors.
Then we can write 
\begin{equation} \label{eq:vec_det_factor}
\widehat{h} = 
\det(B_1) \cdots \det(B_k) \,
\left(\prod_{j_1}\det(B_{\mathcal{I}_1}^{j_1}), \ldots,
\prod_{j_N}\det(B_{\mathcal{I}_N}^{j_N})
\right)
\end{equation}
Therefore, we can split the problem of classifying homeostasis types supported by $\mathcal{G}$ into two cases according to whether the components of $\widehat{h}$ have a common scalar factor or not.

\begin{definition} \label{definition_classification_homeostasis_pleio_coinc} \normalfont
Let $\mathcal{G}$ be a core multiple inputs network and consider its vector determinant $\widehat{h}$ \eqref{eq:vec_det_factor}.
A homeostasis block corresponding to scalar factor $\det(B_i)$ of $\widehat{h}$ is called a \emph{pleiotropic homeostasis block}.
The other homeostasis blocks of are called \emph{coincidental}.
\END
\end{definition}

\begin{remark}\normalfont 
In genetics, pleiotropy refers to the phenomenon when a single locus affects multiple traits. Here, we employed the term \textit{pleiotropic homeostasis} referring to the fact that the vanishing of one single homeostasis block leads to the vanishing of the whole homeostasis vector $\widehat{h}$.
\END
\end{remark}

\begin{definition} \normalfont
Let $\mathcal{G}$ be a core multiple inputs network.
\begin{enumerate}[label=(\alph*)]
\item We say that \emph{pleiotropic homeostasis} occurs when at least one pleiotropic block has vanishing determinant at some fixed input value.
The pleiotropic blocks determine the \emph{pleiotropic homeostasis types} of $\mathcal{G}$, which can be of appendage, structural or counterweight types.
\item We say that \emph{coincidental homeostasis} occurs when a $N$-tuple of coincidental blocks $\big(B_{\mathcal{I}_1}^{j_1},\ldots, B_{\mathcal{I}_N}^{j_N}\big)$ has simultaneously vanishing determinants at some fixed input value.
The $N$-tuples of coincidental blocks determine the \emph{coincidental homeostasis types} of $\mathcal{G}$. 
Coincidental blocks can be of appendage, structural or counterweight types. 
Hence, the coincidental homeostasis type associated to a $N$-tuple of coincidental blocks is a combination of single homeostasis types.
\END
\end{enumerate}
\end{definition}

\subsection{Reduction to the Single Input Node Case}

As discussed above, although the algorithm for the  classification of homeostasis subnetworks is formulated for single input node networks, the other cases can be reduced to this case, as far as the classification of homeostasis subnetworks is concerned.
That is, in applying the algorithm, one does not need to specify the detailed dynamics (knowing at least one such dynamics exists is enough); only the network structure (i.e., the adjacency matrix or edge list) is needed for the algorithm to determine the classification. 

Networks with multiple input nodes, $\GG$, can be readily adapted to the single input node case by augmenting the original network to $\GG^\diamond$ through the introduction of a new input node that connects to all original input nodes (Theorem \ref{thm:GenH}).
For multiple inputs networks, one can apply the algorithm repeatedly by considering one input parameter at a time, augmenting each specialized network with a new input node if necessary, and then combine the results together as explained above.
The reduction of the specialized networks to their core network is done automatically by the algorithm (see Section \ref{S:algorithm-overview}).
As shown before, this step is necessary to obtain only the subnetworks corresponding to the irreducible factors of the vector determinant.

\section{Algorithm Overview}
\label{S:algorithm-overview}

We develop an automated procedure for identifying and classifying all homeostasis subnetworks in input-output networks, based on the theoretical framework introduced in Section \ref{sec:homeostasis-theory}. 
This automated classification algorithm operates on a single user-defined \emph{input data file}, with a special structure, that encodes the network topology (represented as a directed graph) together with the designation of input and output nodes. 
To specify the network topology, users may provide either an admissible system associated with the network structure, the corresponding adjacency matrix, or an edge list specifying network connectivity. 

\begin{remark} \label{rem:multi-input}\normalfont
The algorithm is designed for single-input-single-output networks with single input node. For networks with multiple input nodes, users must first reformulate the network into an equivalent single-input-node representation, as described in Subsection \ref{sec:siso_mi} (Theorem \ref{thm:GenH}). 
For networks with multiple inputs, users must provide separate input data files for each designated input parameter and run the algorithm independently for each case, then combine the resulting outputs, as described in Subsection \ref{sec:miso}. 
We illustrate this procedure using concrete examples in the following section.
\END
\end{remark}

Given an input data file defining an input-output network, the algorithm identifies all homeostasis subnetworks $\KK_\eta$, including appendage homeostasis subnetworks and structural homeostasis subnetworks, as outlined in algorithm \ref{sec:algorithm}. 
For each $\KK_\eta$, the algorithm derives the corresponding homeostasis condition $\det(B_\eta)=0$, where $B_\eta$ is an irreducible diagonal block $B_\eta$ of the homeostasis matrix $H$ for the full input-output network. 

Below, we illustrate the workflow of the algorithm and its typical outputs using a 12-node example network (Fig.~\ref{fig:12-network_diagram}) whose homeostatic structure has been fully analyzed in \cite{WHAG21}. All tables and figures in this section, except Fig. \ref{fig:12-network_diagram}, are generated directly from the algorithm.

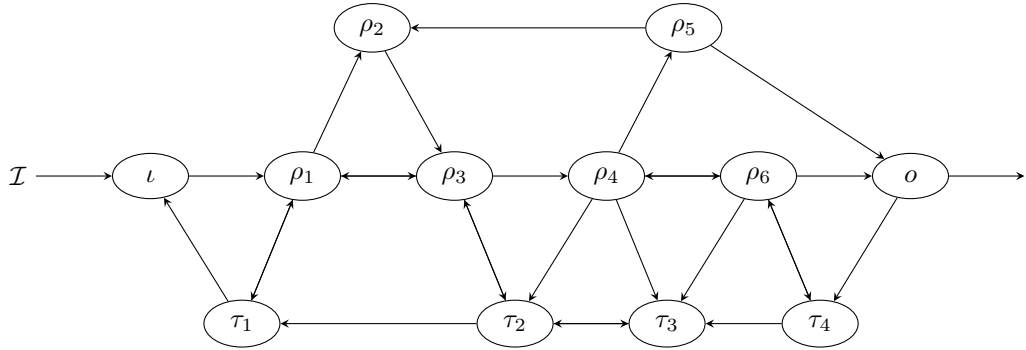
\begin{figure}[!ht]
    \centering
    \begin{tikzpicture}[node distance=2cm and 1cm]
        \node (iota) [process] {$\iota$};
        \node (I) [left=1cm of iota] {$\II$};
        \node (rho1) [process, right of=iota] {$\rho_1$};
        \node (tau1) [process, below left=1.5cm and 0.08cm of rho1] {$\tau_1$};
        \node (rho2) [process, above right=1.5cm and 0.2cm of rho1] {$\rho_2$};
        \node (rho3) [process, right of=rho1] {$\rho_3$};
        \node (rho4) [process, right of=rho3] {$\rho_4$};
        \node (tau2) [process, below right = 1.5cm and 0.08cm of rho3] {$\tau_2$};
        \node (rho5) [process, above right=1.5cm and 0.3cm of rho4] {$\rho_5$};
        \node (rho6) [process, right of=rho4] {$\rho_6$};
        \node (o) [process, right of=rho6] {$o$};
        \node (tau3) [process, right of=tau2] {$\tau_3$};
        \node (tau4) [process, right = 1.0cm of tau3] {$\tau_4$};

        \draw [arrow] (I) -- (iota);
        \draw [arrow] (iota) -- (rho1);
        \draw [arrow] (rho1) -- (rho2);
        \draw [arrow] (rho1) -- (rho3);
        \draw [arrow] (rho1) -- (tau1);
        \draw [arrow] (rho2) -- (rho3);
        \draw [arrow] (rho3) -- (rho1);
        \draw [arrow] (rho3) -- (rho4);
        \draw [arrow] (rho3) -- (tau2);
        \draw [arrow] (rho4) -- (tau2);
        \draw [arrow] (rho4) -- (tau3);
        \draw [arrow] (rho4) -- (rho5);
        \draw [arrow] (rho4) -- (rho6);
        \draw [arrow] (rho5) -- (rho2);
        \draw [arrow] (rho5) -- (o);
        \draw [arrow] (rho6) -- (tau4);
        \draw [arrow] (rho6) -- (o);
        \draw [arrow] (rho6) -- (tau3);
        \draw [arrow] (rho6) -- (rho4);
        \draw [arrow] (tau1) -- (iota);
        \draw [arrow] (tau1) -- (rho1);
        \draw [arrow] (tau2) -- (tau1);
        \draw [arrow] (tau2) -- (rho3);
        \draw [arrow] (tau2) -- (tau3);
        \draw [arrow] (tau3) -- (tau2);
        \draw [arrow] (tau4) -- (tau3);
        \draw [arrow] (tau4) -- (rho6);
        \draw [arrow] (o) -- (tau4);
        \draw [arrow] (o) -- ++(1.5cm, 0);

    \end{tikzpicture}
    \caption{Network Diagram of the 12-node example, adapted from \cite{WHAG21}.}
    \label{fig:12-network_diagram}
\end{figure}

\subsection{12-node Network Example}

In the 12-node network example (Fig.~\ref{fig:12-network_diagram}), input node $\iota$ receives a single input $\mathcal{I}$ and node $o$ is the output node. 
For computational purposes, the algorithm represents network nodes with integer IDs. In this example, nodes are assigned IDs $0$ through $11$, as summarized in Table~\ref{tab:variable_map_12node}. 

\begin{table}[h!]
\centering
\begin{tabular}{clcl}
\toprule
Node ID & Node Label & Node ID & Node Label \\
\midrule
0  & $\iota$ (input node)  & 6  & $\rho_6$  \\
1  & $\rho_1$     & 7  & $\tau_1$    \\
2  & $\rho_2$     & 8  & $\tau_2$   \\
3  & $\rho_3$       & 9  & $\tau_3$     \\
4  & $\rho_4$     & 10  & $\tau_4$     \\
5  & $\rho_5$  & 11 & $o$ (output node)   \\ 
\bottomrule
\end{tabular}
\caption{Mapping between internal node IDs and original node labels for the 12-node example.}
\label{tab:variable_map_12node}
\end{table}

An input data file for this network is shown in Fig.~\ref{fig:12_node_input}. The network structure is specified through an edge list representation, where a pair $(i,j)$ represents an arrow from node $i$ to $j$. The first line lists the node IDs used in the network, while the next two lines designate the input and output nodes. 

\begin{figure}[!htp]
    \centering
\includegraphics[width=\linewidth]{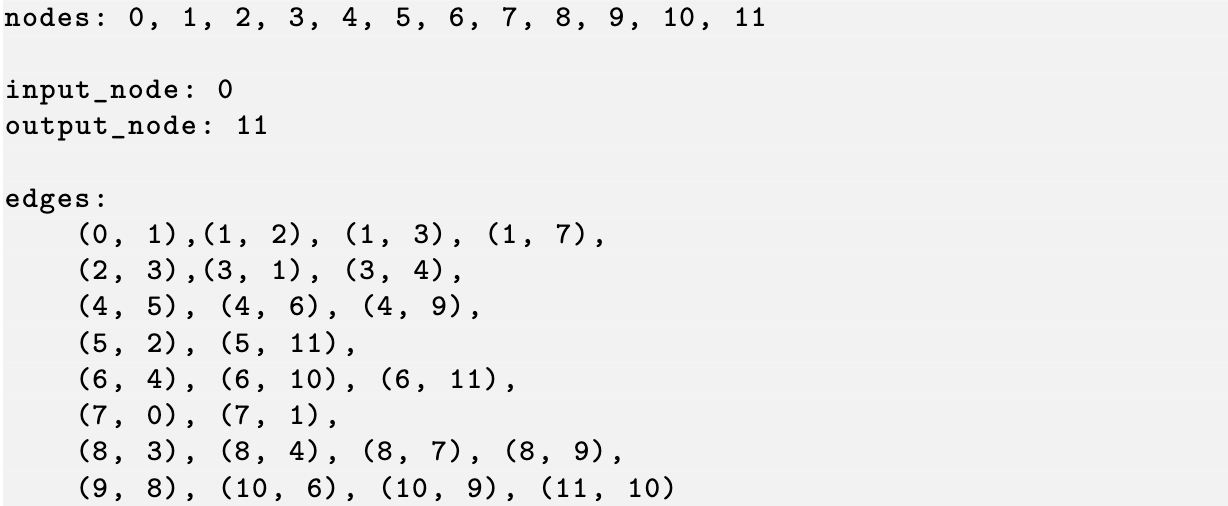}
    \caption{Input data file for the 12-node network example}
    \label{fig:12_node_input}
\end{figure}



Given this input file, the algorithm constructs the corresponding input-output network and extracts the \emph{core network}, as shown in Fig.~\ref{fig: 12-node-core}. 

\begin{figure}[!ht]
    \centering
\includegraphics[width=0.6\textwidth]{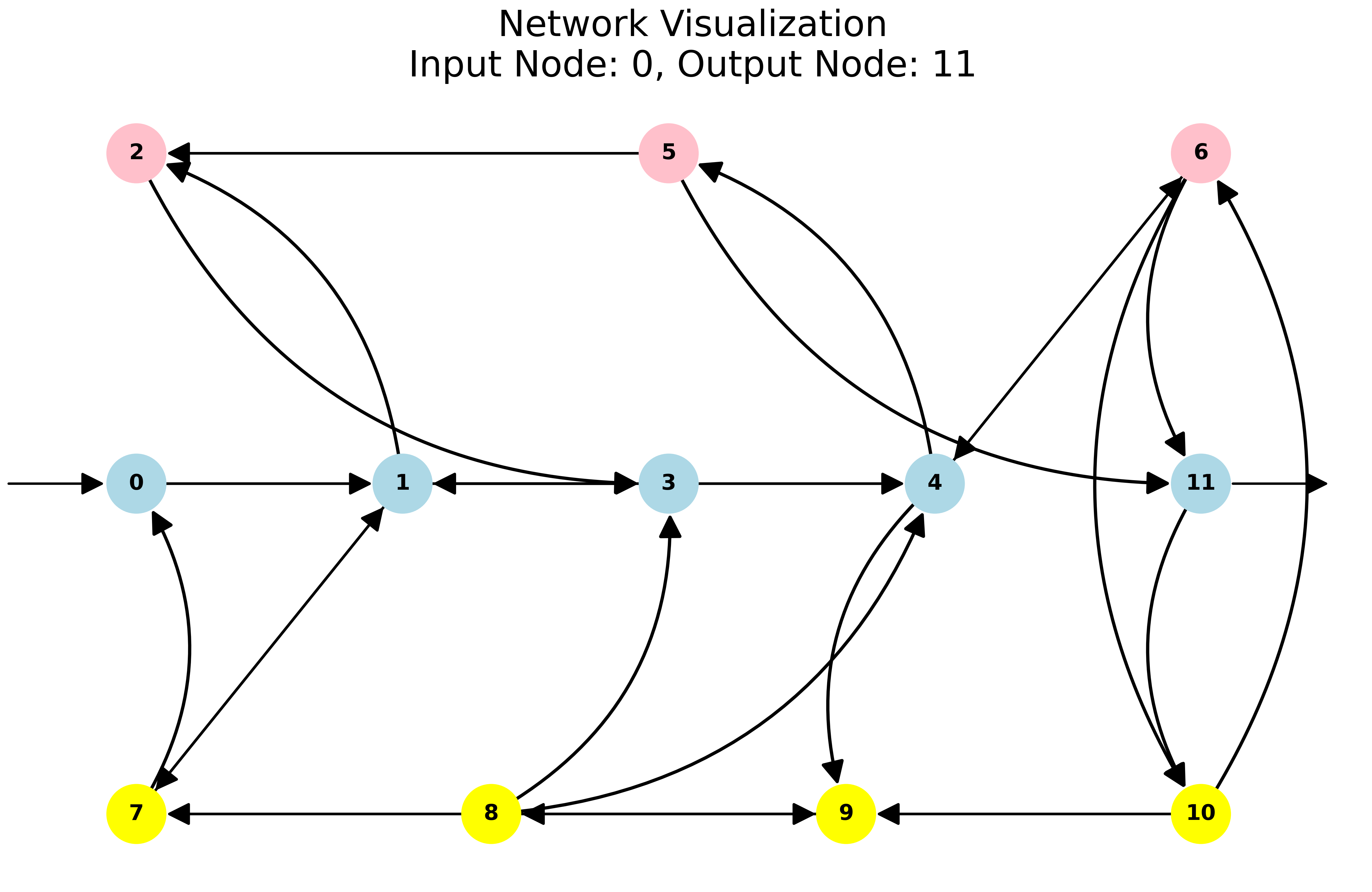}
\caption{Algorithm-generated core network for the 12-node example from Fig.~\ref{fig:12-network_diagram}. Blue: super-simple nodes; pink: simple but non-super-simple nodes; yellow: appendage nodes. }
    \label{fig: 12-node-core}
\end{figure}

Users may optionally provide a label mapping in the input data file to associate each node ID with a user-defined node label (e.g., see Fig.~\ref{fig:Cholesterol_input_output_network} for an input data file of the Cholesterol example in Subsection \ref{sec:chol}). When such a mapping is supplied, all exported tables and network visualizations are generated using these labels (see Fig.~\ref{fig:Cholesterol_output12}).

\subsection{Simple paths and node categories}

Within the core input-output network (Fig.~\ref{fig: 12-node-core}), the algorithm enumerates all \emph{$\iota o$-simple paths} from the designated input node (node $0$) to the designated output node (node $11$), together with their associated \emph{complementary subnetworks}. Based on these outputs, nodes are classified into \emph{super-simple nodes}, \emph{simple nodes}, and \emph{appendage nodes}. The \emph{appendage subnetwork} $\AA_\GG$ consisting of all appendage nodes and their induced edges is also constructed. These outputs are returned by the algorithm in tabular form; some are shown in Fig.~\ref{fig: Simple_Paths_and_Complementary_Subnetworks}.

\begin{figure}[htbp!]
    \centering  \includegraphics[width=\textwidth]{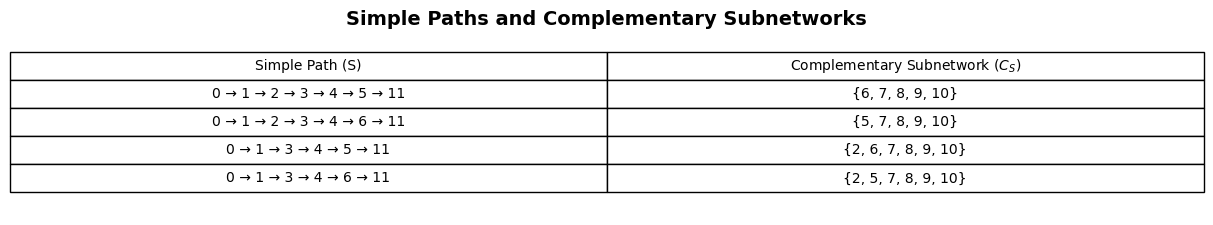}\\    \includegraphics[width=\textwidth]{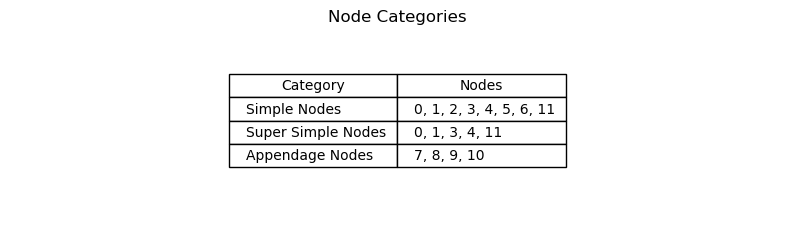}
    \caption{Algorithm-generated tables of simple paths, complementary subnetworks, and node categories for the 12-node example. In the node categories table, super-simple nodes are arranged in downstream order: $0>1>3>4>11$.}
    \label{fig: Simple_Paths_and_Complementary_Subnetworks}
\end{figure}

\begin{remark} \normalfont
In all algorithm-generated network visualizations (e.g., Fig.~\ref{fig: 12-node-core}), blue nodes denote super-simple nodes, pink nodes denote simple but not super-simple nodes, and yellow nodes denote appendage node.
\END
\end{remark}

In addition to network visualizations and tables, the algorithm writes its textual output to an automatically generated file, \texttt{run\_output.txt}. In what follows, boxed blue text reproduces selected excerpts from the output file for the 12-node example.

\subsection{Appendage homeostasis subnetworks}

As a first step in classifying homeostasis subnetworks, the algorithm determines the \emph{appendage path components} from the appendage subnetwork $\AA_\GG$:
\begin{tcolorbox}[colframe=blue!80!black, colback=white, fontupper=\color{blue}\ttfamily]
Appendage Path Components:\\
\{7\}, \{8, 9\}, \{10\}
\end{tcolorbox}

Among these components, those satisfying the no-cycle condition are classified as the appendage homeostasis subnetworks $\mathcal{A}_j$.
\begin{tcolorbox}[colframe=blue!80!black, colback=white, fontupper=\color{blue}\ttfamily]
Appendage Homeostasis Subnetworks (Appendage Components that Satisfy the No-Cycle Condition):\\
Component \{8, 9\}\\
Component \{7\}
\end{tcolorbox}

These homeostasis subnetworks are also returned as network visualizations, shown in Fig.~\ref{fig: 12-node_appendage_homeostasis_subnetworks}.

\begin{figure}[htbp!]
    \centering
    \includegraphics[width=0.7\textwidth]{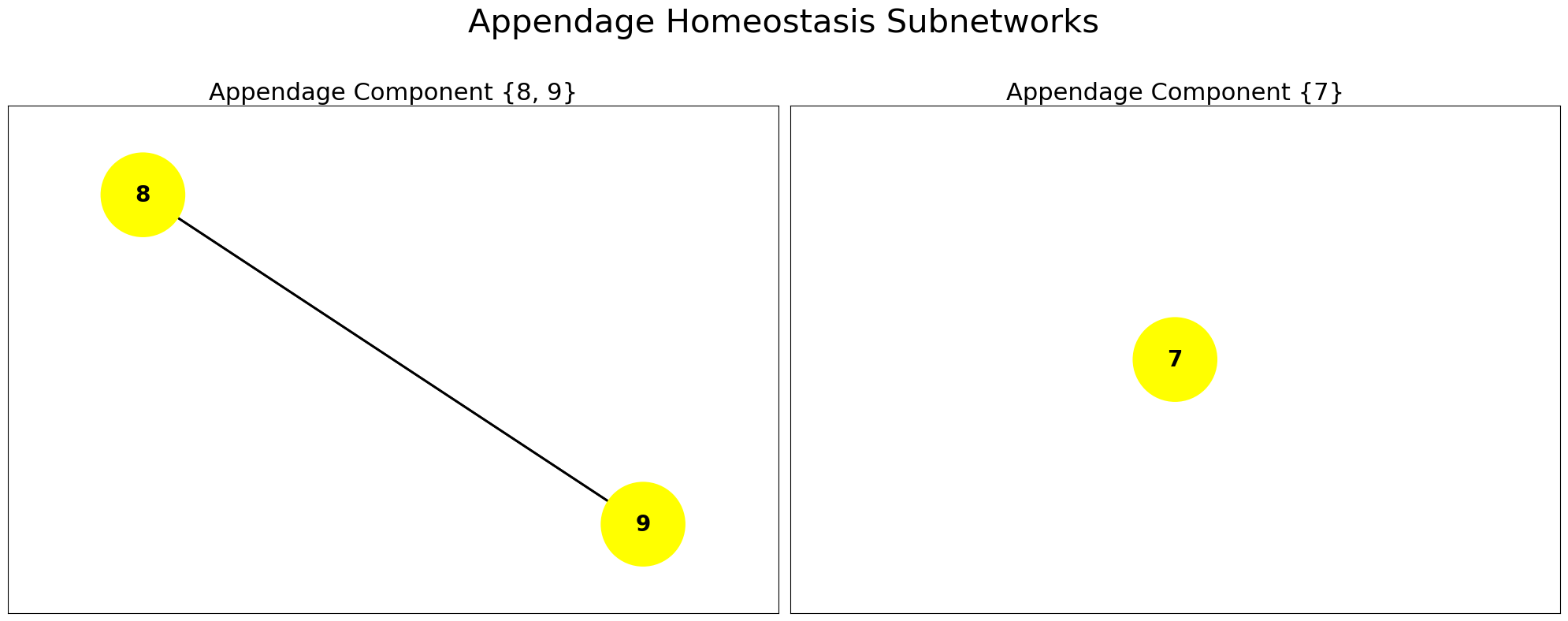}
    \captionsetup{justification=centering}
    \caption{Algorithm-generated appendage homeostasis subnetworks for the 12-node example.}
    \label{fig: 12-node_appendage_homeostasis_subnetworks}
\end{figure}

\subsection{Structural homeostasis subnetworks}

The algorithm next determines the structural homeostasis subnetworks by enumerating all
\emph{super-simple structural subnetworks} constructed from adjacent super-simple nodes. Since there are five super-simple nodes (Fig.~\ref{fig: Simple_Paths_and_Complementary_Subnetworks}), this yields a total of four structural homeostasis subnetworks:
\begin{tcolorbox}[colframe=blue!80!black, colback=white, fontupper=\color{blue}\ttfamily]
Super-Simple Subnetwork: \(\mathcal{L}'(0,\, 1) = \{0,\, 1\}\)\\
Super-Simple Subnetwork: \(\mathcal{L}'(1,\, 3) = \{1,\, 2,\, 3\}\)\\
Super-Simple Subnetwork: \(\mathcal{L}'(3,\, 4) = \{3,\, 4\}\)\\
Super-Simple Subnetwork: \(\mathcal{L}'(4,\, 11) = \{4,\, 5,\, 6,\, 10,\, 11\}\)
\end{tcolorbox}

\begin{figure}[htbp!]
    \centering
    \includegraphics[width=0.6\textwidth]{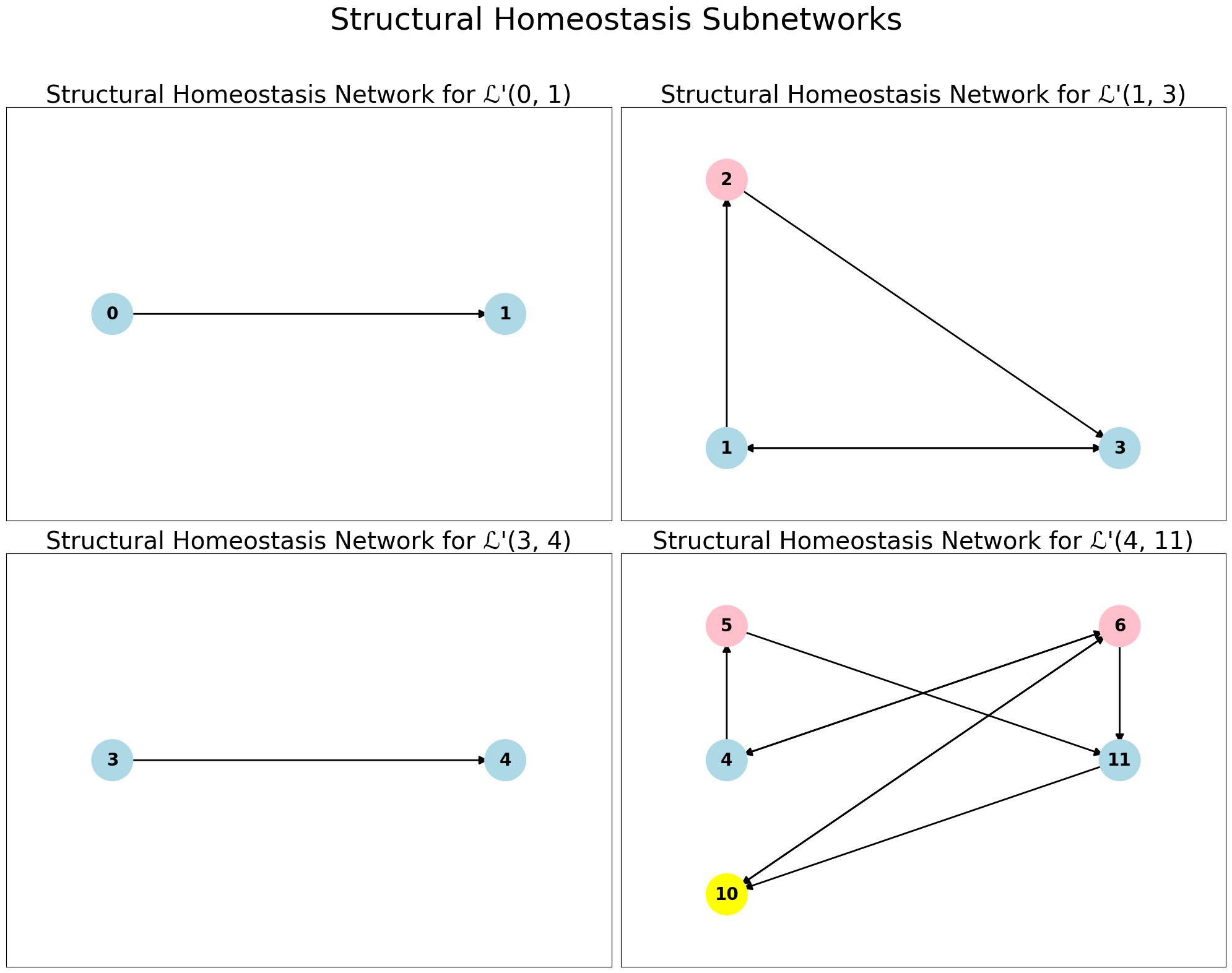}
    \captionsetup{justification=centering}
    \caption{Algorithm-generated structural homeostasis subnetworks for the 12-node example. }
    \label{fig:12-node-structural-homeostasis-subnetworks}
\end{figure}

Visualizations of these structural homeostasis subnetworks are shown in Fig.~\ref{fig:12-node-structural-homeostasis-subnetworks}.

Finally, the algorithm returns a table summarizing all identified homeostasis subnetworks, shown in Fig.~\ref{fig:12-node-homeostasis-subnetworks-table}. 
\begin{figure}[htbp!]
    \centering
    \includegraphics[width=\textwidth]{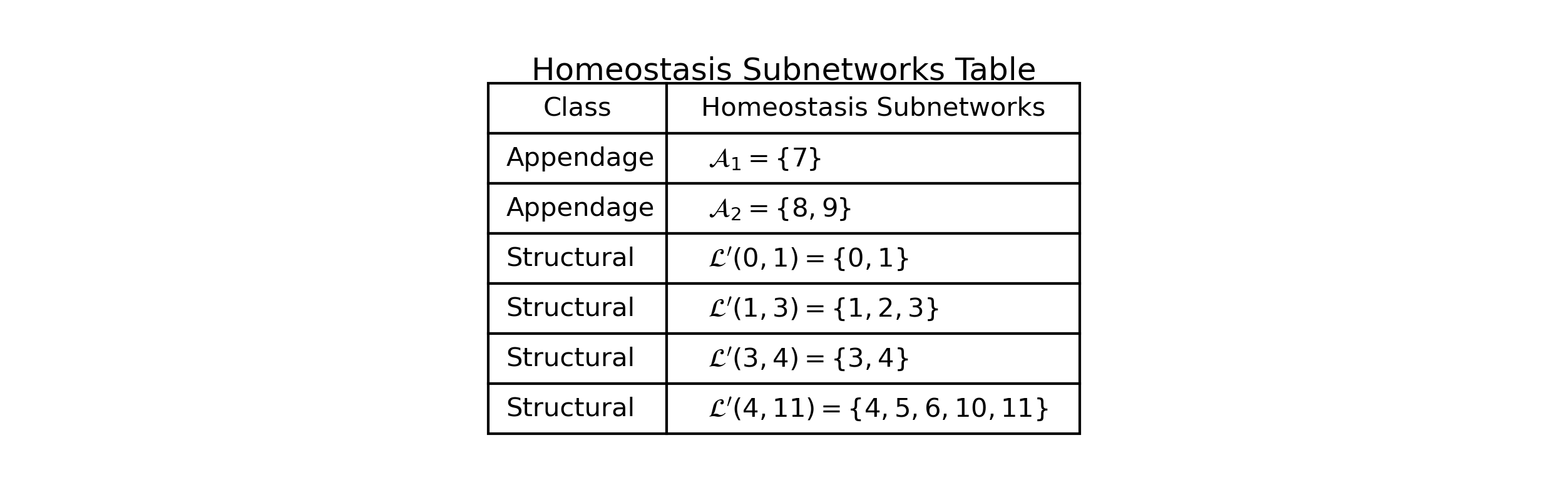}
    \caption{Algorithm-generated table of homeostasis subnetworks for the 
    12-node example.}
    \label{fig:12-node-homeostasis-subnetworks-table}
\end{figure}

For each identified homeostasis subnetwork ($\KK_\eta$), the algorithm also outputs the associated irreducible homeostasis block ($B_\eta$) and saves them to \texttt{run\_output.txt}, as listed below. The corresponding homeostasis condition is $\det(B_\eta)=0$ and $\det(B_\xi)\neq 0$ for all $\xi\neq \eta$. For example, to verify whether homeostasis in the 12-node example is induced by $\AA_2$ (Fig.~\ref{fig:12-node-homeostasis-subnetworks-table}), one evaluates the corresponding matrix block at the equilibrium and checks whether its determinant $f_{8, x_8} f_{9, x_9} - f_{8, x_9} f_{9, x_8}$ vanishes at the homeostasis point. 

\begin{tcolorbox}[colframe=blue!80!black, colback=white, fontupper=\color{blue}\ttfamily]
Homeostasis block for $\AA_1$:
\[
\begin{bmatrix}
f_{7,x_7}
\end{bmatrix}
\]

Homeostasis block for $\AA_2$:
\[
\begin{bmatrix}
f_{8,x_8} & f_{8,x_9} \\
f_{9,x_8} & f_{9,x_9}
\end{bmatrix}
\]

Homeostasis block for $\mathcal{L}'(0,1)$:
\[
\begin{bmatrix}
f_{1,x_0}
\end{bmatrix}
\]

Homeostasis block for $\mathcal{L}'(1,3)$:
\[
\begin{bmatrix}
f_{2,x_1} & f_{2,x_2} \\
f_{3,x_1} & f_{3,x_2}
\end{bmatrix}
\]

Homeostasis block for $\mathcal{L}'(3,4)$:
\[
\begin{bmatrix}
f_{4,x_3}
\end{bmatrix}
\]

Homeostasis block for $\mathcal{L}'(4,11)$:
\[
\begin{bmatrix}
f_{5,x_4}  & f_{5,x_5}  & 0         & 0 \\
f_{6,x_4}  & 0          & f_{6,x_6}  & f_{6,x_{10}} \\
0          & 0          & f_{10,x_6} & f_{10,x_{10}} \\
0          & f_{11,x_5} & f_{11,x_6} & 0
\end{bmatrix}
\]

\end{tcolorbox}

The results obtained by the algorithm for this example are fully consistent with the analytic classification reported in \cite{WHAG21}, validating the accuracy of our algorithmic framework.

\section{Applications}\label{sec:application}

In this section, we apply our algorithm to a series of biologically motivated network models, ranging from small to relatively large networks and covering both single-input and multiple-input network architectures. We consider both networks in which the input and output nodes are distinct and networks in which the input coincides the output. Our goal is to automatically identify all homeostasis motifs that would be difficult or impractical to detect manually, particularly in large networks. For networks of small to intermediate size, the algorithm results are validated against analytic classification results available from previous studies or carried out in the present work.

\subsection{Cholesterol network with single input node}\label{sec:chol}

Our first example is a model of intracellular cholesterol regulation proposed by \cite{pool2018integrated}. This is a single-input-single-output network consisting of 12 nodes. In the corresponding \emph{input data file} (Fig.~\ref{fig:Cholesterol_input_output_network}), we assign IDs $1$ through $12$ to the network nodes (i.e., model variables), as summarized in Table \ref{tab:variable_map_cholesterol}. 

\begin{table}[h!]
\centering
\resizebox{\textwidth}{!}{%
\begin{tabular}{clp{4.5cm}clp{4.5cm}}
\toprule
Nodes & Variables & Biological name & Nodes & Variables & Biological name \\
\midrule
1  & $m_h$ & HMGCR mRNA    & 7  & $v_E$ (input node)  & free VLDL \\
2  & $m_r$  & LDLR mRNA  & 8  & $v_{RB}$ &receptor bound VLDL\\
3  & $h$   &  HMGCR  & 9  & $v_I$ &  internalised
VLDL  \\
4  & $l_E$  & free LDL   & 10 & $r_f$ &  free unbound receptors  \\
5  & $l_{RB}$& receptor bound LDL & 11 & $r_I$  &  internalised receptors \\
6  & $l_I$  & internalised LDL  & 12 & $c$ (output node)   & intracellular cholesterol \\
\bottomrule
\end{tabular}}
\caption{Node IDs, variable names and biological meanings for the cholesterol model. LDL: low density lipoprotein; VLDL: very low density lipoprotein.}
\label{tab:variable_map_cholesterol}
\end{table}

Based on the model formulation in \cite{pool2018integrated}, we first write down the generic admissible system (see \eqref{eq:cholesterol} in Appendix \ref{app:chol}). The corresponding mathematical network structure can be encoded in the input data file using the admissible system, its adjacency matrix or edge list. Here we use an edge list representation (see Fig.~\ref{fig:Cholesterol_input_output_network}). Below this list, a mapping from node IDs to their corresponding variable names is also included. The algorithm results are therefore reported using these labels rather than numeric node IDs (see Fig.~\ref{fig:Cholesterol_output12}). 

\begin{figure}[htbp!]
    \centering
    \includegraphics[width=0.5\textwidth]{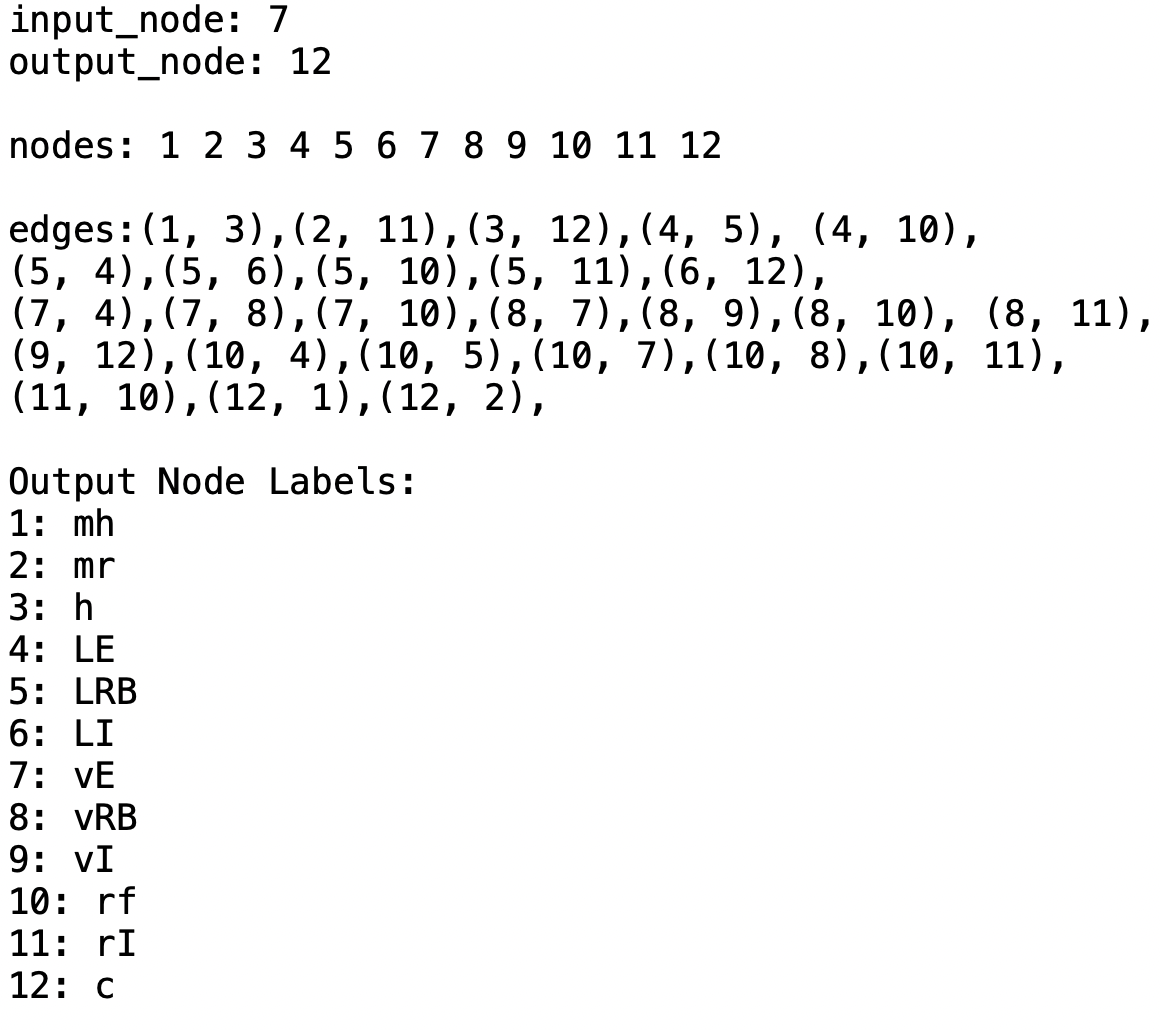}
    \caption{Input data file for the cholesterol network.}
\label{fig:Cholesterol_input_output_network}
\end{figure}

Fig.~\ref{fig:Cholesterol_output12} includes a visualization of the cholesterol network with node $v_E$ designated as the input node and node $c$ as the output node, a table of $\iota o$-simple paths, a table of node categories, and the final classification table listing all identified homeostasis subnetworks. Consistent with the analytic classification results in \cite{duncan2024homeostasis}, our algorithm identifies the same homeostasis subnetworks: one structural homeostasis motif $\LL'\{v_E, c\}$ (see its network visualization in Fig.~\ref{fig:Cholesterol_output12}, lower left), and three appendage subnetworks, each consisting of a single appendage node. 

To determine which of these mechanisms underlies homeostasis of the output node $c$ in the full biochemical model, the corresponding homeostasis conditions ($\det(B_\eta)=0$) must be checked. The algorithm lists all the irreducible homeostasis blocks $B_\eta$ in the output file. We also provide them in Appendix \ref{app:chol}, Fig.~\ref{fig:cholesterole-homeostasis-conditions}. 

\begin{figure}[t!]
\centering
\includegraphics[width=\textwidth]{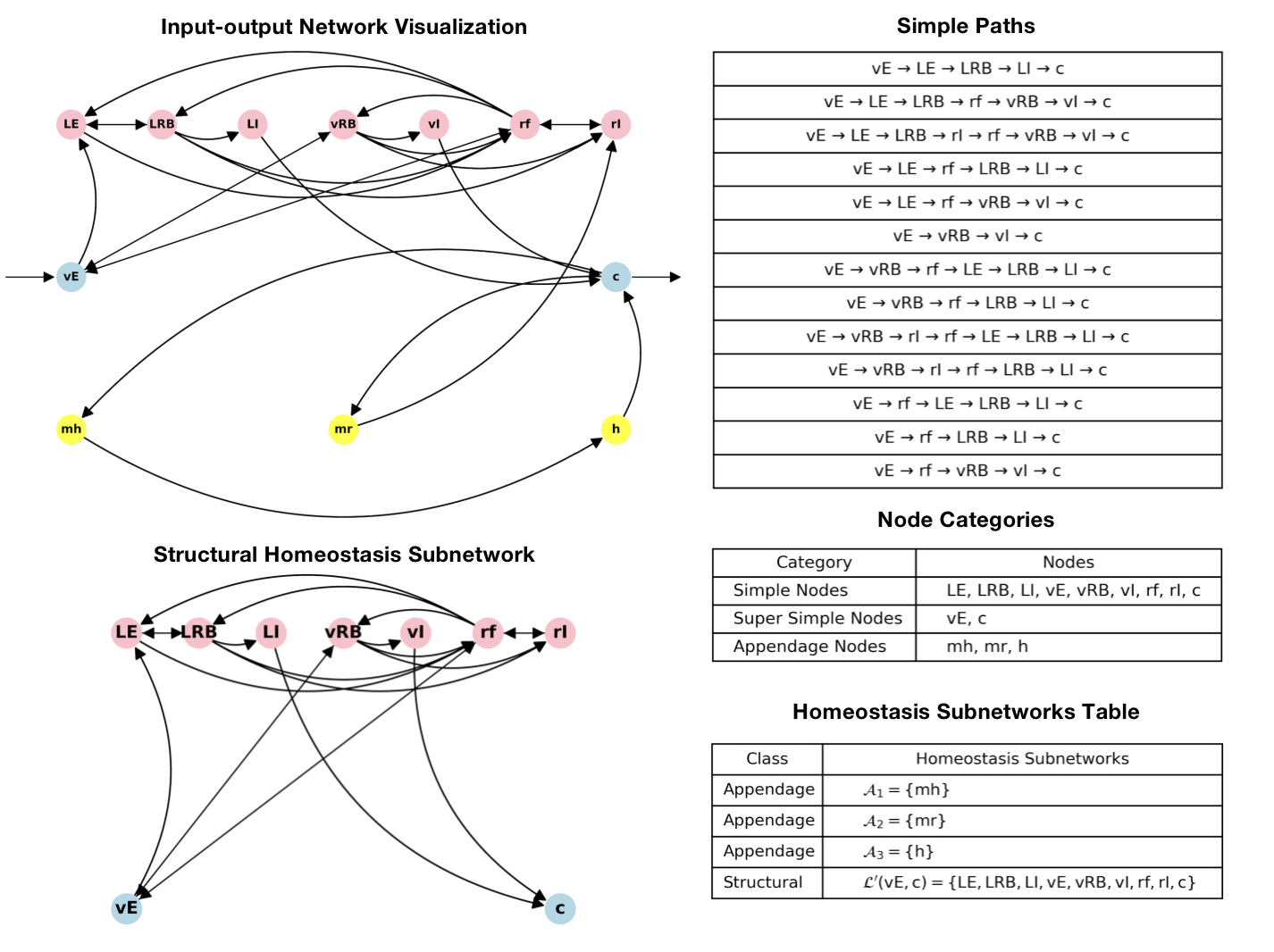}
\caption{Algorithm-generated outputs for the cholesterol model. Blue: super-simple nodes; pink: simple but non-super-simple nodes; yellow: appendage nodes.}
\label{fig:Cholesterol_output12}
\end{figure}

\begin{remark} \normalfont
\label{rm:conservation}
For demonstration purposes, we consider admissible systems under the assumption that all variables in the original ODE system (in this example and all subsequent examples) are independent. 
Biochemical network models, however, frequently exhibit conservation laws or other algebraic constraints. 
When such constraints are present, the system should first be reduced to an equivalent ODE formulation involving only independent variables (see \cite{oellerich2021,jin2026}). 
For example, if a conservation relation of the form $x_1+x_2=C$ holds, where $C$ is a constant, one eliminates a dependent variable (e.g., by substituting $x_1=C-x_2$) and writes the admissible system in terms of the remaining independent variables. 
The algorithm then applies to the reduced system without modification.
\END
\end{remark}

\subsection{Chemotaxis network with multiple input nodes}\label{sec:chemotaxis}

The mathematical modeling of chemotaxis can be roughly divided into two types: single cell models and bacterial population models~\cite{tindall08a,tindall08b}.
Single cell models consider the activation of the flagellar motor by detection of attractants and repellents in the extracellular medium. 
The flagellar motor activity of bacteria is regulated by a signal transduction pathway, which integrates changes of environmental chemical concentrations into a behavioral response. 
Assuming mass-action kinetics, the reactions in the signal transduction pathway can be modeled mathematically by ODEs. 
The population models describe evolution of bacterial density by parabolic PDEs allowing movement up-the-gradient, the most prominent feature of chemotaxis. 

Understanding the response of bacteria such as \textit{E. coli} to external attractants has been the subject of experimental work and mathematical models for nearly 40 years. 
Many models of the chemotaxis have been formulated and developed to provide a comprehensive description of the cellular processes and include details of receptor methylation, ligand-receptor binding and its subsequent effect on the biochemical signaling cascade, along with a description of motor driving CheY/CheY-P levels, the main output of the chemotaxis system  (see~\cite{tindall08a} for a survey).

Here, we consider a minimal model for the \textit{E. coli} response proposed by \cite{clausznitzer10,edgington15,edgington18}, which nevertheless is in good agreement with experimental findings. 
It has four variables for the concentrations of CheA/CheA-P ($a_p$), CheY/CheY-P ($y_p$), CheB/CheB-P ($b_p$) and the receptor methylation ($m$) and is given by the following system of ODEs (in non-dimensional form):
\begin{equation} \label{original_e_coli}
\begin{aligned}
& \frac{d m}{d t} = \gamma_{R}\, (1 - \phi(m,L)) - \gamma_{B} \, \phi(m,L) \, b_{p}^{2} \\
& \frac{d a_{p}}{d t} = \phi(m,L) \, k_{1} \, (1 - a_{p}) - k_{2} \, (1 - y_{p})\, a_{p} - k_{3} \, (1 - b_{p}) \, a_{p} \\
& \frac{d y_{p}}{d t} = \alpha_{1} \, k_{2} \, (1 - y_{p})a_{p} - k_{4} \, y_{p} \\
& \frac{d b_{p}}{d t} = \alpha_{2} \, k_{3} \, (1 - b_{p}) \, a_{p} - k_{5} \, b_{p}
\end{aligned}
\end{equation}
where $\gamma_B$, $\gamma_R$, $k_1,\ldots,k_5$ are non-dimensional parameters, the extracellular ligand concentration $L$ is the \emph{external parameter} and the function $\phi$ is determined by a  Monod–Wyman–Changeux (MWC) description of receptor clustering
\begin{equation} \label{definition_phi}
    \phi (m,L) = \frac{1}{1 + e^{F(m,L)}} 
    \qquad\text{with}\qquad 
    F(m,L) = N\left[ 1 - \frac{m}{2} + \log \left( \frac{1 + \frac{L}{K_{a}^{\textrm{off}}}}{1 + \frac{L}{K_{a}^{\textrm{on}}}}\right) \right]
\end{equation}

The key observation of~\cite{edgington18} is that system \eqref{original_e_coli} has a unique asymptotically stable equilibrium $X^*=(m^*,a_p^*,y_p^*,b_p^*)$, with $a_p^*$, $y_p^*$ and $b_p^*$ positive and $m^{*}$ is a real number, for the non-dimensional parameters obtained from the parameter values originally used in~\cite{clausznitzer10}. 
Furthermore, \cite{edgington18} were able to show that some pairs of parameters might yield oscillatory behavior, but in regions of parameter space that are outside the experimentally relevant.

\begin{figure}[!ht]
\centering
\includegraphics[width=0.4\linewidth]{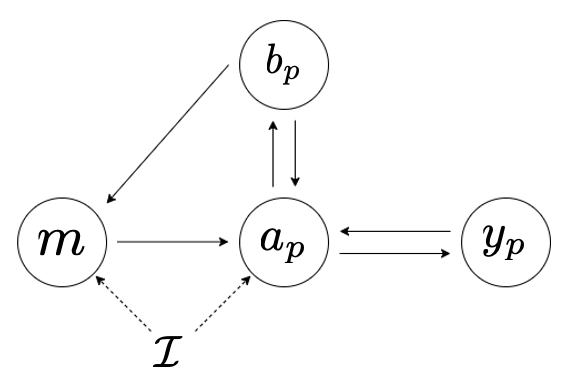} \qquad
\includegraphics[width=0.4\linewidth]{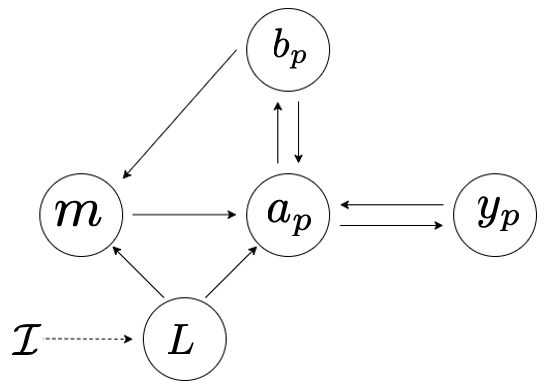}
\caption{\label{f:e_coli_network} (Left ) Network $\GG$ corresponding to the model equations \eqref{original_e_coli} for \emph{E. coli} chemotaxis.
(Right) Augmented network $\GG^\diamond$ obtained from $\GG$.}
\end{figure}

The network corresponding to the model equations \eqref{original_e_coli} has two input nodes $m$ and $a_p$ (Fig.~\ref{f:e_coli_network}, left).
In Table \ref{tab:variable_map_chemotaxis} we list the variables
of model \eqref{original_e_coli} and the corresponding nodes in the network.
As described in Remark \ref{rem:multi-input}, in order to apply the algorithm, one must first convert this multiple-input-node network into a single-input-node representation. To do so, we introduce a new input node $L$ and connect it to both $m$ and $a_p$. This converts the original multiple-input-node network $G$ to the augmented single-input-node network $G^\diamond$ (Fig.~\ref{f:e_coli_network}, right). By Theorem \ref{thm:GenH}, this reformulation preserves both the homeostasis matrix and the homeostasis subnetworks of the original network system. 

\begin{table}[!ht]
\centering
\begin{tabular}{cl l}
\toprule
Nodes & Variables & Biological name \\
\midrule
1 & $L$ (input node) & Extracellular ligand\\
2 & $m$ & Receptor methylation \\
3 & $a_p$ & CheA/CheA-P \\
4 & $b_p$ & CheB/CheB-P\\
5 & $y_p$ (output node) & CheY/CheY-P\\
\bottomrule
\end{tabular}
\caption{Node IDs, variable names and biological meanings for the chemotaxis model.}
\label{tab:variable_map_chemotaxis}
\end{table}

Representative outputs obtained by the classification algorithm applied to the augmented network $G^\diamond$ are shown in Fig.~\ref{fig:Chemotaxis_example}, with a layout similar to that in Fig.~\ref{fig:Cholesterol_output12}. The algorithm identifies three homeostasis mechanisms, summarized in the lower-right table. Among them is a structural homeostasis motif, $\LL'\{L, a_p\}=\{L,m,a_p\}$, arising from the balance between the two paths $L\to m\to a_p$ and $L\to a_p$. A visualization of this motif is shown on the right. The other two mechanisms correspond to substrate inhibition along the edge $a_p\to y_p$ and null-degradation on the appendage node $b_p$. 
Going back to the model equations \eqref{original_e_coli} and using the homeostasis condition returned by the algorithm (not shown), it can be checked that the structural homeostasis motif, $\LL'\{L, a_p\}$, is the one responsible for the homeostasis observed in this model.

In the original approach of \cite{madeira2022homeostasis} the homeostasis subnetworks are computed by hand, using the method described in the paper.
The result of the calculation is that there are three types of homeostasis: (1) appendage (null-degradation) homeostasis associated with $b_p$, (2) structural (Haldane) homeostasis associated with $a_p\to y_p$ and (3) input counterweight homeostasis associated with $m \to a_p$.
This is consistent with the result obtained here.
The input counterweight homeostasis subnetwork $m \to a_p$ corresponds to the structural block $\LL'\{L, a_p\}=\{L,m,a_p\}$, when the node $L$ is removed from the network.

\begin{figure}[htbp!]
    \centering
    \includegraphics[width=\textwidth]{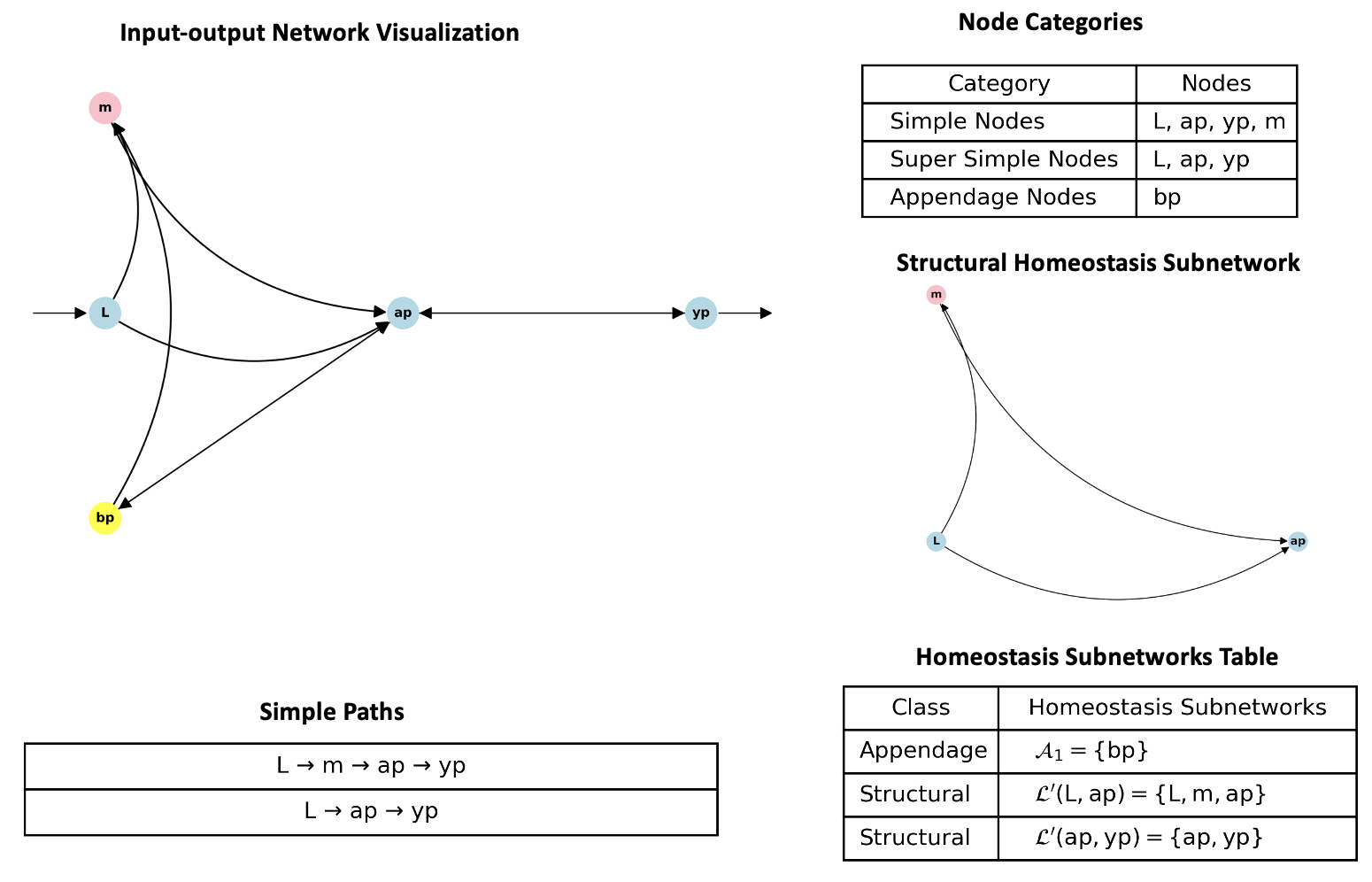}
    \caption{Algorithm-generated outputs for the chemotaxis model. Blue: super-simple nodes; pink: simple but non-super-simple nodes; yellow: appendage nodes.}
    \label{fig:Chemotaxis_example}
\end{figure}

\subsection{Dopamine network with multiple input nodes}\label{subsec:dopamine}

Best et al. \cite{best2009homeostatic,NBR14} proposed a mathematical model to study homeostasis of extracellular dopamine (eDA) in response to variations in the activities of the enzyme tyrosine hydroxylase (TH) and the dopamine transporters (DATs). This model is given by a system of differential equations describing the underlying biochemical network (see, \cite{best2009homeostatic}, Figure 1).
Their modeling and simulation results show that eDA remains approximately constant over a range of TH and DAT activities, indicating the capability of the network to maintain eDA homeostasis. 

\begin{table}[!htp]
\centering
\resizebox{\textwidth}{!}{%
\begin{tabular}{clp{4.5cm}clp{4.5cm}}
\toprule
Nodes & Variables & Biological name & Nodes & Variables & Biological name \\
\midrule
1  & TH (input node)  & Tyrosine hydroxylase    & 6  & cDa & Cytosolic dopamine \\
2  & bh2  &  Dihydrobiopterin & 7  & vDa& Vesicular dopamine  \\
3  & bh4  &  Tetrahydrobiopterin & 8 & DAT & Dopamine transporter\\
4  & tyr & Tyrosine & 9  & tyrpool & The tyrosine pool \\
5  & L-dopa & 3,4-dihyroxyphenylalanine &  10 & eDa (output node) & Extracellular dopamine \\
\bottomrule
\end{tabular}}
\caption{Mapping of node IDs to variable names and biological meanings for the dopamine model.}
\label{tab:variable_map_dopamine}
\end{table}

\begin{figure}[htbp!]
\centering
\includegraphics[width=\textwidth]{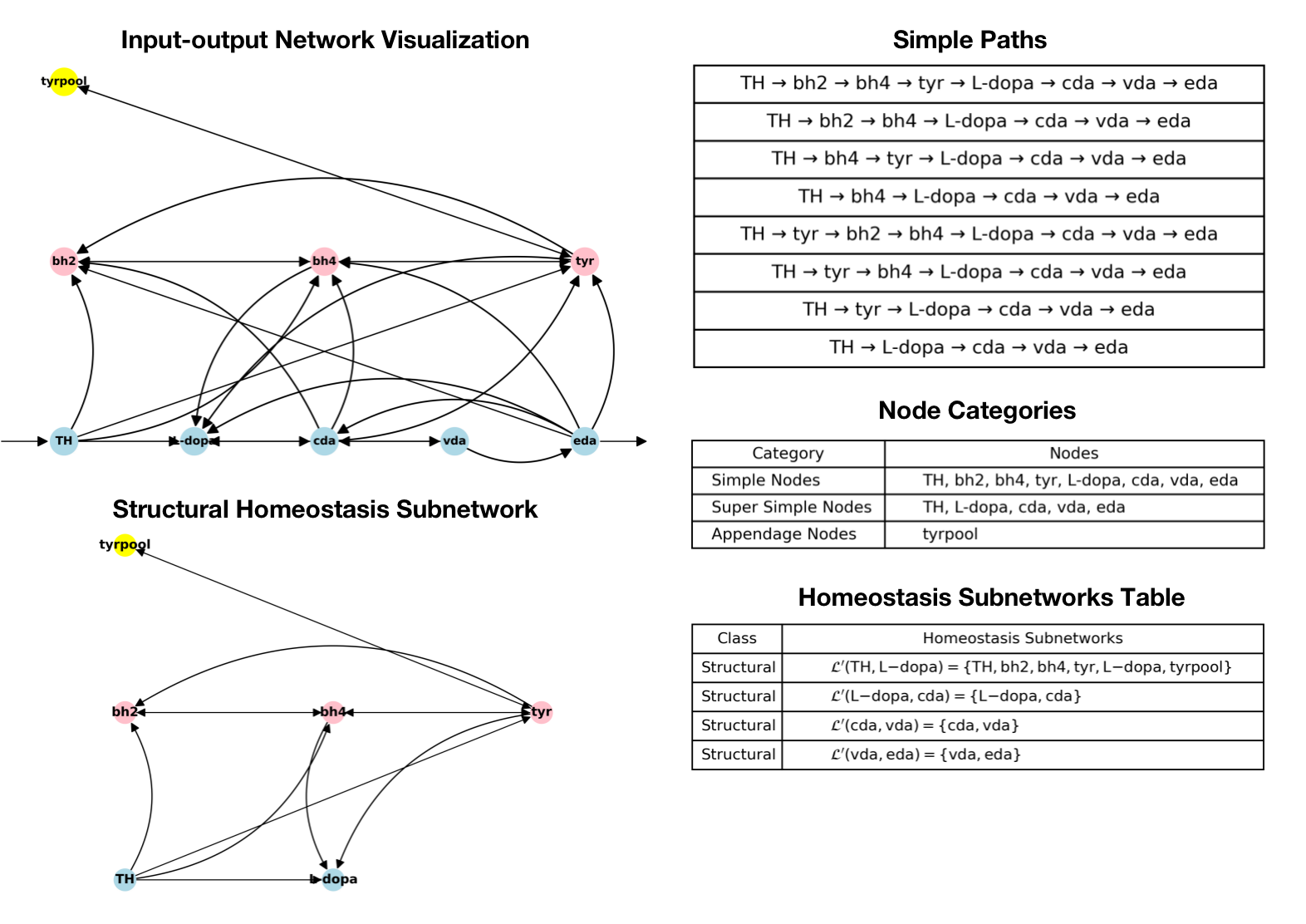}
\caption{Algorithm-generated outputs for the dopamine model. Blue: super-simple nodes; pink: simple but non-super-simple nodes; yellow: appendage nodes.}
\label{fig: Dopamine_Example_Output}
\end{figure}

Here, we focus on the case in which TH activity is treated as the input parameter and apply our algorithm to identify homeostatic mechanisms embedded in the full biological network structure. Since variations in TH activity influence four variables, bh2, bh4, tyr and L-dopa, the corresponding network is a single-input-single-output network $G$ with multiple input nodes. As discussed in Subsection \ref{sec:chemotaxis}, applying the algorithm requires converting $G$ into an augmented network $G^\diamond$ by introducing a new input node TH, and connecting it to nodes bh2, bh4, tyr and L-dopa. The core network of $G^\diamond$ is shown in Fig.~\ref{fig: Dopamine_Example_Output} (upper left). Note that the node DAT is excluded from the core network. Table \ref{tab:variable_map_dopamine} lists the model variables, their biological names, and the corresponding nodes in $G^\diamond$.

Representative outputs obtained from applying the algorithm to $G^\diamond$ are shown in Fig. \ref{fig: Dopamine_Example_Output}, with a similar layout as in Fig.~\ref{fig:Cholesterol_output12}.
For a general admissible system respecting the structure of the dopamine biological network, infinitesimal homeostasis of eDA can arise through four distinct mechanisms, summarized in the lower-right table. Three of these correspond to substrate inhibition along edge L-dopa $\to$ cDa, cDa $\to$ vDa, or vDA $\to$ eDa. The fourth homeostasis mechanism arises through a larger structural subnetwork $\mathcal{L}'\{\text{TH},\text{L-dopa}\}$ consisting of nodes TH, bh2, bh4, tyr, L-dopa and tyrpool. A visualization of this subnetwork is displayed in the lower-left panel of Fig. \ref{fig: Dopamine_Example_Output}. The homeostasis blocks corresponding to these four mechanisms are provided in Appendix \ref{app:dopamine-block}, Fig.~\ref{fig:dopamine-homeostasis-conditions}.

To validate the outputs of our algorithm, we also compute the Jacobian matrix and the homeostasis matrix $H$ of the dopamine admissible system \eqref{eq:DA-admissible}, and directly determine the irreducible factors of $\det(H)$ whose vanishing corresponds to homeostasis mechanisms (see Appendix \ref{app:dopa-compute}). These results are consistent with the algorithm outputs.

\subsection{Hepatic one-carbon metabolism coupled with methionine, choline and betaine synthesis: multiple inputs and multiple choices of output nodes}
\label{sec:sexdiff}

In this subsection, we apply the algorithm to a 17-node input-output biochemical network underlying the model of hepatic one-carbon metabolism proposed in \cite{sadre2018sex}. 
We examine homeostasis mechanisms in choline (cho) and homocysteine (hcy) under variations in input parameters methionine and methylenetetrahydrofolate reductase (MTHFR) activity. 
This network contains two distinct input parameters and two possible output nodes. 
The remaining state variable (network node) names and their biological meanings are listed in Table \ref{tab:variable_map_sex_difference}. 

\begin{table}[!ht]
\centering
\resizebox{\textwidth}{!}{%
\begin{tabular}{clp{4.5cm}clp{4.5cm}}
\toprule
Nodes & Variables & Biological name & Nodes & Variables & Biological name\\
\midrule
1 & met & Methionine  & 10 & mthf &5-methyltetrahydrofolate   \\
2 & sam & S-adenosylmethionine    & 11 & gnmt & Glycine N-methyltransferase    \\
3 & sah & S-adenosylhomocysteine      & 12 & gnmtf & GNMT-5mTHF   \\
4 & hcy & Homocysteine      & 13 & fgnmtf &  5mTHF-GNMT-5mTHF  \\
5 & dhf & Dihydrofolate      & 14 & bet& Betaine  \\
6 & thf & Tetrahydrofolate      & 15 & bet-ald & Betaine aldehyde \\
7 & fthf & 10-formyltetrahydrofolate     & 16 & cho & Choline     \\
8 & ch &  5,10-methenyltetrahydrofolate     & 17 & pc & PtCho      \\
9 & ch2 & 5,10-methylenetrahydrofolate    &    &         \\
\bottomrule
\end{tabular}}
\caption{Mapping of node IDs to variable names and biological meanings for the hepatic one-carbon metabolism model.}
\label{tab:variable_map_sex_difference}
\end{table}

Variations in the first input methionine only affects a single input node \texttt{met}. 
We can therefore directly apply our algorithm and classify all homeostasis mechanisms with either \texttt{cho} or \texttt{hcy} as the output node. Representative outputs for these cases are shown in Figs.~\ref{fig:sex_difference_example_1_16} and \ref{fig:sex_difference_example_1_4}. 
In contrast, the other input parameter, MTHFR activity, influences two state variables \texttt{ch2} and \texttt{mthf}.
This leads to a multiple-input node network $G$ and, as discussed previously, can be handled by introducing a new input node denoted as \texttt{MTHFR} and adding edges from this node to \texttt{ch2} and \texttt{mthf} (see the augmented input-output networks $G^\diamond$ in Fig.~\ref{fig:sex_difference_example_MTHFR}). 

\begin{figure}[!htbp]
\centering
\includegraphics[width=\textwidth]{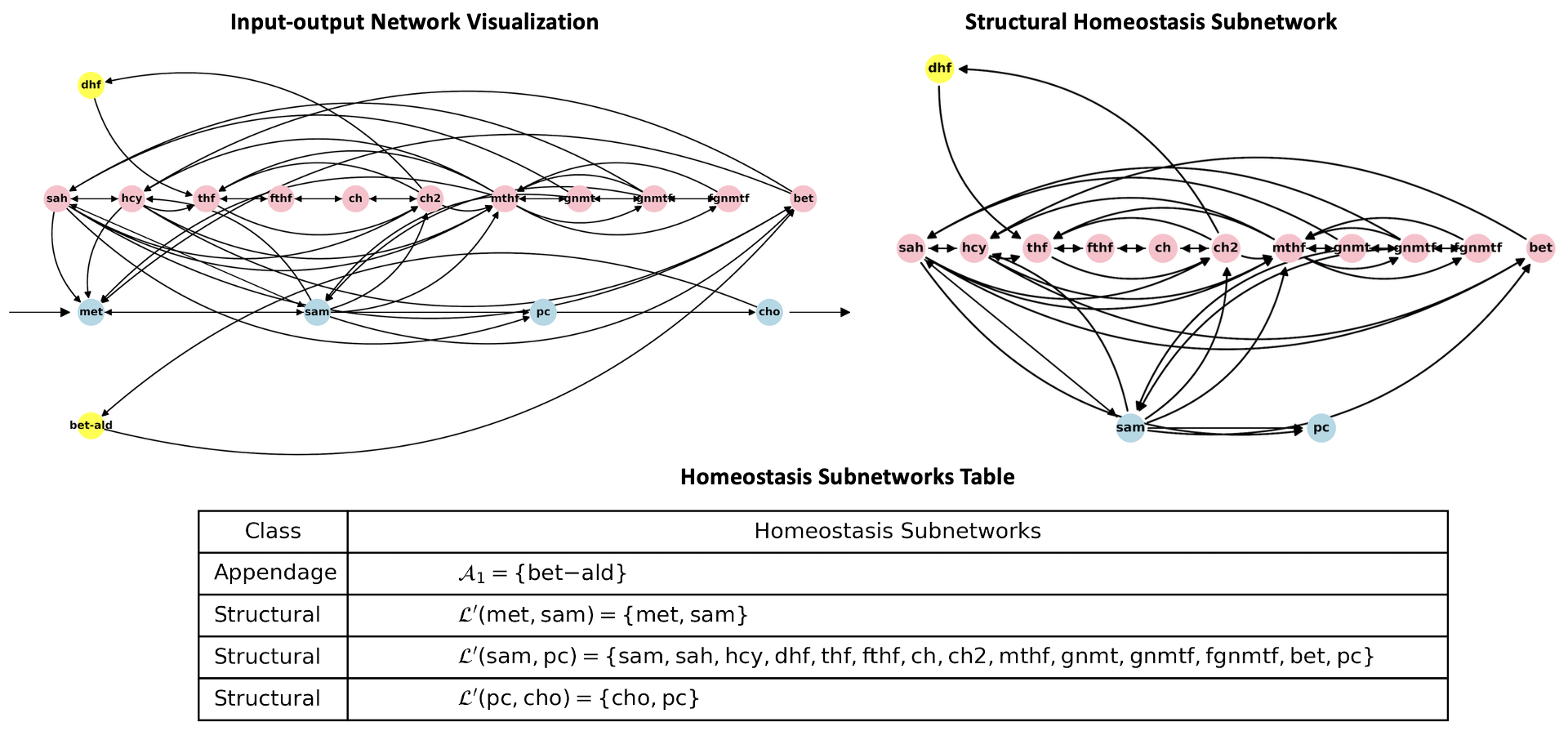}
\caption{Algorithm-generated outputs for the one-carbon metabolism model with input node \texttt{met} and output node \texttt{cho}. Blue: super-simple nodes; pink: simple but non-super-simple nodes; yellow: appendage nodes.
There are 54 $\iota o$-simple paths in total.}
\label{fig:sex_difference_example_1_16}
\end{figure}

When \texttt{met} is the input node and \texttt{cho} is the output (Fig.~\ref{fig:sex_difference_example_1_16}), there are 54 $\iota o$-simple paths. The algorithm identifies four homeostasis mechanisms: one appendage subnetwork consisting of a single homeostasis node \texttt{bet-ald}, and three structural homeostasis subnetworks. Two of the structural mechanisms arise from substrate inhibition along either the edge \texttt{met}$\to$\texttt{sam} or the edge \texttt{pc}$\to$\texttt{cho}. The third is a 14-node structural homeostasis subnetwork $\mathcal{L}'\{\texttt{sam},\texttt{pc}\}$, whose structure is displayed in Fig.~\ref{fig:sex_difference_example_1_16} (upper right).   

\begin{figure}[!htbp]
\centering
\includegraphics[width=\textwidth]{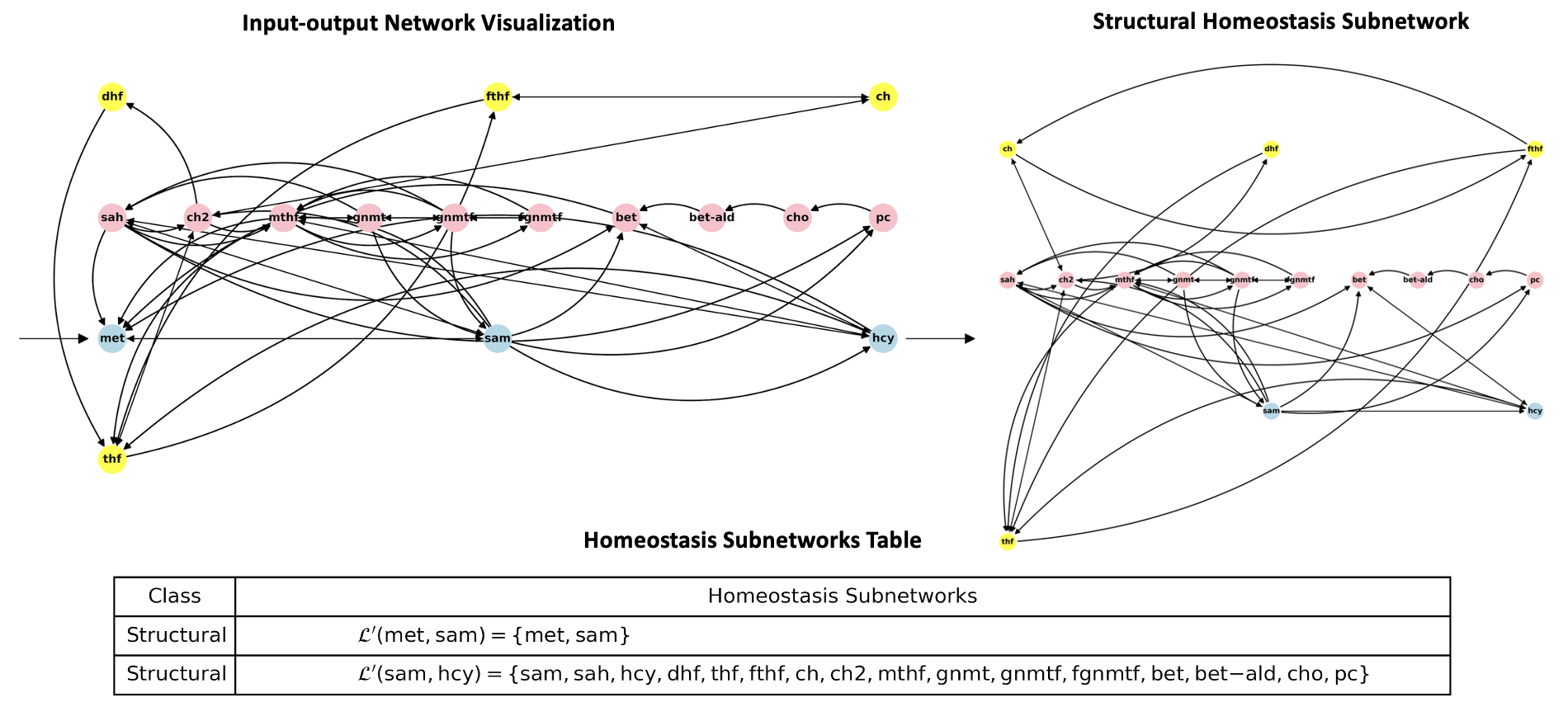}
\caption{Algorithm-generated outputs for the one-carbon metabolism model with input node \texttt{met} and output node \texttt{hcy}, with a total of $45$ $\iota o$-simple paths. Blue: super-simple nodes; pink: simple but non-super-simple nodes; yellow: appendage nodes.}
\label{fig:sex_difference_example_1_4}
\end{figure}

When \texttt{met} is the input node and \texttt{hcy} the output (Fig.~\ref{fig:sex_difference_example_1_4}), there are 45 $\iota o$-simple paths. In this case, the algorithm detects two structural homeostasis mechanisms. Infinitesimal homeostasis in \texttt{hcy} can arise either via substrate inhibition along \texttt{met}$\to$\texttt{sam}, or through a 16-node structural subnetwork $\mathcal{L}'\{\texttt{sam},\texttt{hcy}\}$, whose network structure is displayed in Fig.~\ref{fig:sex_difference_example_1_4} (upper right). 


When the input node is \texttt{MTHFR} and output node is \texttt{cho}, the algorithm detects four homeostasis mechanisms (see Fig.~\ref{fig:sex_difference_example_MTHFR}A): one appendage and three structural. 
In particular, two blocks ($\AA_1=\{\texttt{bet-ald}\}$ and $\LL'\{\texttt{pc, cho}\}$) are identical to those identified when \texttt{met} is the input node. These correspond to pleiotropic homeostasis types, in which the vanishing of either block leads to homeostasis in \texttt{cho} with respect to both methionine input and MTHFR activity
(see Definition \ref{definition_classification_homeostasis_pleio_coinc}). 
By contrast, when \texttt{MTHFR} is the input node and \texttt{hcy} is the output node, the algorithm identifies two structural homeostasis mechanisms (see Fig.~\ref{fig:sex_difference_example_MTHFR}B). Neither coincides with those obtained when \texttt{met} is the input node.
Therefore, in this case there is only coincidental homeostasis.

\begin{figure}[!htbp]
\centering
\includegraphics[width=\textwidth]{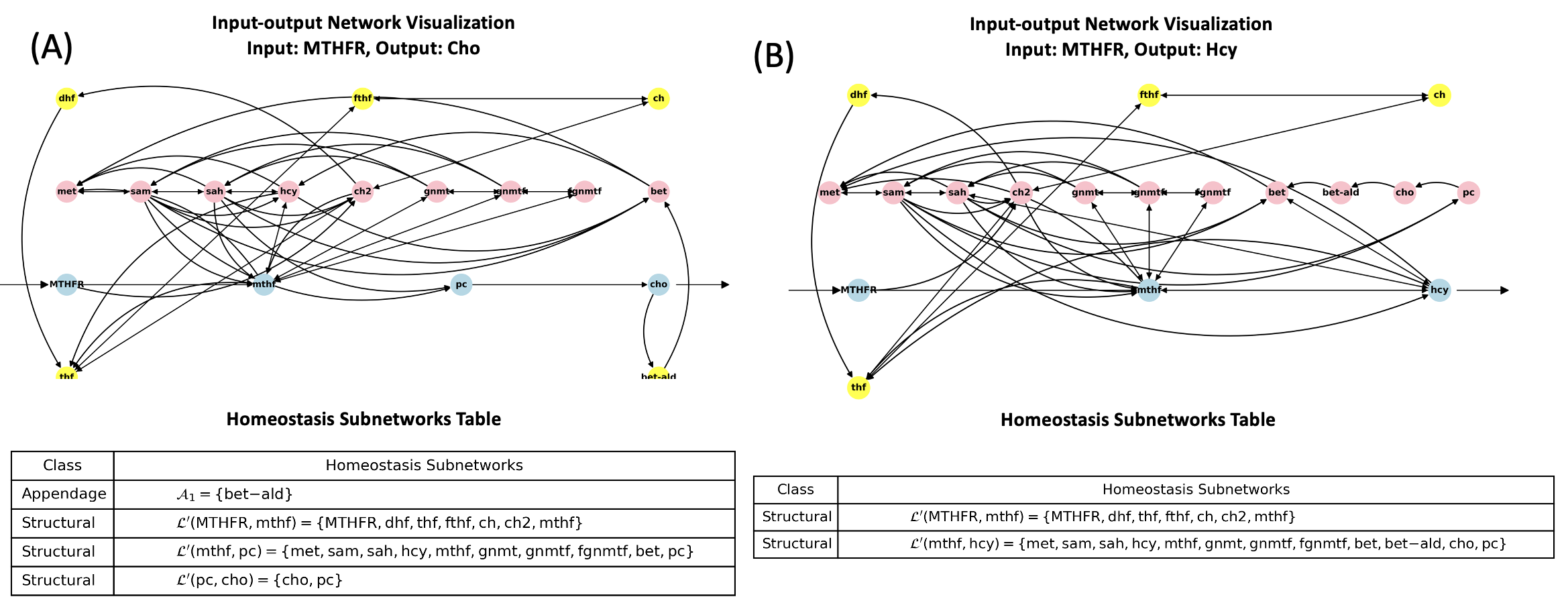}
\caption{Algorithm-generated outputs for the one-carbon metabolism model with the input node \texttt{MTHFR} and output node (A) \texttt{cho} and (B) \texttt{hcy}. Blue: super-simple nodes; pink: simple but non-super-simple nodes; yellow: appendage nodes.}
\label{fig:sex_difference_example_MTHFR}
\end{figure}


\subsection{Hepatic one carbon metabolism coupled with glutathione synthesis pathway: moderate size network with multiple input node choices}

In this subsection, we apply the algorithm to a 34-node input-output biochemical network underlying the model of hepatic one-carbon metabolism coupled with the glutathione synthesis pathway proposed in 
\cite{reed2008mathematical}.
This is an extension of the model for hepatic mitochondrial folate metabolism from \cite{nijhout2006}, designed to include cysteine and glutathione metabolism. 
This mathematical model is quite
complicated, since it includes all of one carbon metabolism and not just the
transsulfuration pathway in order to couple with the glutathione synthesis pathway.

The 34 variables (nodes) are listed in Table~\ref{tab:variable_map_glutathione}.
The output node is node 14 (\texttt{bglu}). 
We consider two choices of input node and input parameter:
\begin{itemize}
\item Node 13 (\texttt{5mf}). 
Here the input parameter $\II$ is total concentration of cellular folate (one carbon metabolism), see Figure \ref{fig:Glutathione_input13} for the corresponding algorithm outputs. 
\item Node 17 (\texttt{cglu}).
Here, the input parameter $\II$ is the total concentration of cellular glutamate (glutamate synthesis pathway), see Figure \ref{fig:Glutathione_input41} for the corresponding altorighm outputs. 
\end{itemize}

In both cases, there is only one single homeostatic mechanism, given by a large structural block that involves all nodes in the core network. 

\begin{remark} \normalfont
In Figs.~\ref{fig:Glutathione_input13} and \ref{fig:Glutathione_input41}, due to the large size and density of the network, many arrows overlap and are therefore not clearly visible. 
In both cases, node 19 (\texttt{bcys}), node 31 (\texttt{src}), and node 32 (\texttt{dmg}) are excluded from the core network.
\END
\end{remark}

\begin{table}[!ht]
\centering
\resizebox{\textwidth}{!}{%
\begin{tabular}{clp{4.5cm}clp{4.5cm}}
\toprule
Nodes & Variables & Biological name & Nodes & Variables & Biological name \\
\midrule
1  & mthf & tetrahydrofolate & 18 & c10f & 10-formyltetrahydrofolate\\
2  & m10f & 10-formyltetrahydrofolate & 19 & bcys & blood cysteine\\
3  & mser & mitochondrial serine & 20 & cCOO & cytosolic formate\\
4  & mgly & mitochondrial glycine & 21 & cser & cytosolic serine\\
5  & m2cf & 5-10-methylenetetrahydrofolate & 22 & cgly & cytosolic glycine\\
6  & mCOO & mitochondrial formate & 23 & c2cf & 5-10-methylenetetrahydrofolate\\
7  & bgsg & blood glutathione disulfide & 24 & ccys & cytosolic cysteine\\
8  & bgsh & blood glutathione & 25 & sam & S-adenosylmethionine \\
9  & cgsg & cytosolic glutathione disulfide & 26 & sah &  S-adenosylhomocysteine \\
10 & m1cf & 5-10-methenyltetrahydrofolate & 27 & c1cf & 5-10-methenyltetrahydrofolate\\
11 & cthf & tetrahydrofolate & 28 & aic & P-ribosyl-5-amino-4-imidazole carboxamide\\
12 & hcy & homocysteine & 29 & met & methionine \\
13 & 5mf (input node 1) & 5-methyltetrahydrofolate & 30 & cyt & cystathionine \\
14 & bglu (output node) & blood glutamate & 31 & src & sarcosine\\
15 & cgsh & cytosolic glutathione & 32 & dmg & dimethylglycine\\
16 & dhf & dihydrofolate & 33 & bgly & blood glycine\\
17 & cglu (input node 2) & cytosolic glutamate & 34 & glc & glutamyl-cysteine \\
\bottomrule
\end{tabular}}
\caption{Mapping of node IDs to variable names and biological meanings for the glutathione model.}
\label{tab:variable_map_glutathione}
\end{table}

\begin{figure}[!ht]
\centering
\includegraphics[width=0.85\textwidth]{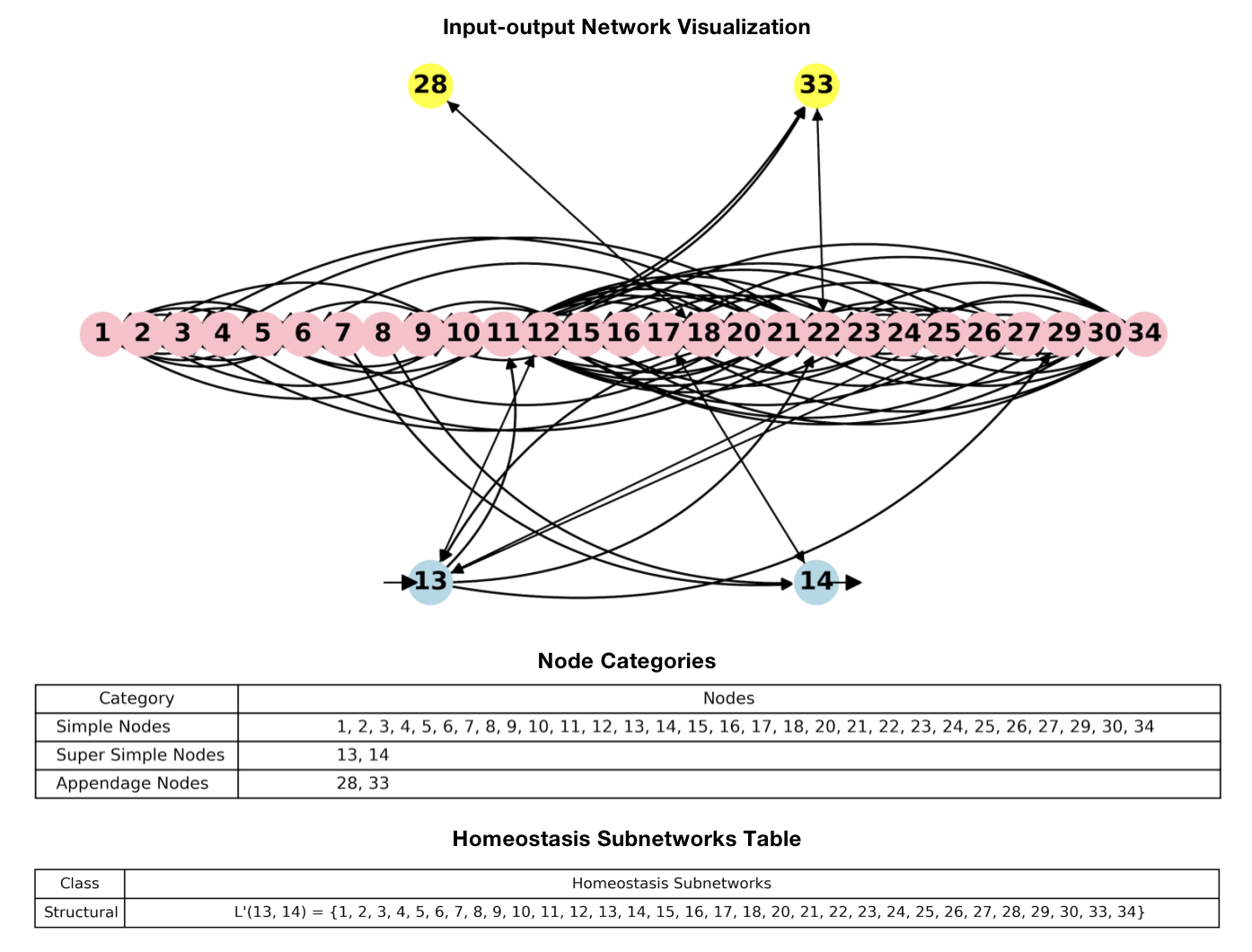}
\caption{Algorithm-generated outputs for the glutathione model with Input Node 13. Blue: super-simple nodes; pink: simple but non-super-simple nodes; yellow: appendage nodes.}
\label{fig:Glutathione_input13}
\end{figure}

\begin{figure}[!ht]
\centering
\includegraphics[width=0.85\textwidth]{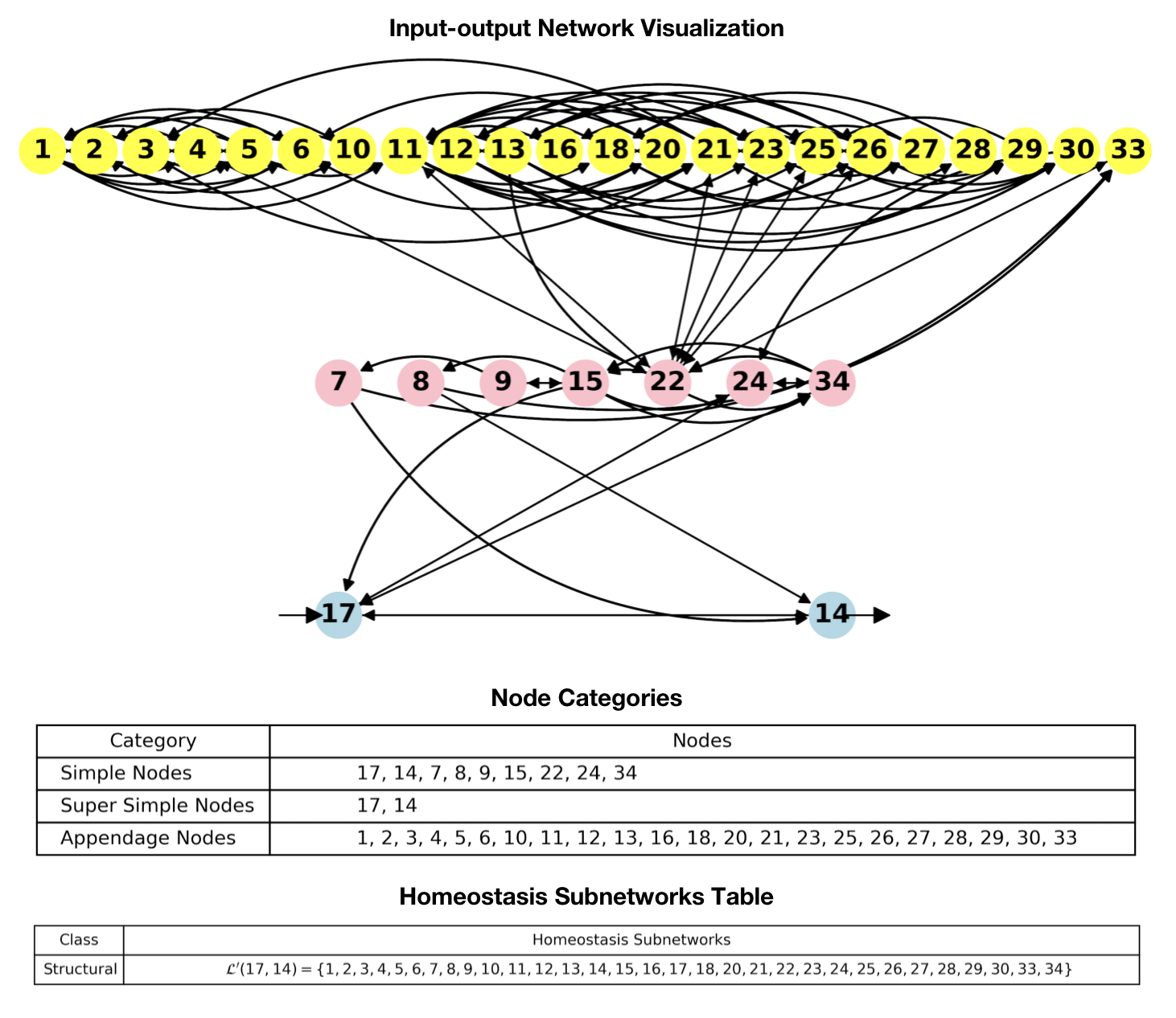}
\caption{Algorithm-generated outputs for the glutathione model with Input Node 17. Blue: super-simple nodes; pink: simple but non-super-simple nodes; yellow: appendage nodes.}
\label{fig:Glutathione_input41}
\end{figure}





\subsection{Zinc homeostasis: network with input=output}
\label{ss:zinc_example}

Zinc is an essential micro-nutrient for plants, because it plays an important role in many enzymes catalyzing vital cellular reactions. 
In higher doses, however, zinc is toxic. 
Therefore, plants have to strictly control and adjust the uptake of zinc through their roots depending on its concentration in the surrounding soil.
This is achieved by a complicated control system consisting of sensors, transmitters and zinc transporter proteins.

In \cite{claus2015} the authors propose a mathematical model for regulation of zinc uptake in roots of \emph{Arabidopsis thaliana} based on the uptake of zinc, expression of a transporter protein and the interaction between an activator and inhibitor.
The equations in \cite{claus2015} are obtained by an \emph{ad hoc} dimensional reduction of a model proposed in \cite{claus2012}.

The model of \cite{claus2015} is given by a system of four nonlinear ordinary differential equations.
For convenience we represent the concentration of each component by: $x_1=[\mathrm{Zn^{2+}}]$, $x_2=[\mathrm{mRNA}]$, $x_3=[\mathrm{ZIP}]$ and $x_4=[\mathrm{Dimer}]$.
That is, $x_1,\ldots,x_4$ are state variables and the equations are
\begin{equation} \label{EQ:ZINC}
\begin{aligned}
\dot{x}_1 & = \mathcal{I}x_4  - c_1 x_1 \\
\dot{x}_2 & = 1 - a_2 x_2 x_3  - c_2 x_2\\
\dot{x}_3 & = a_3 x_1 (v_1 - x_3) -  a_5 x_2 x_3  
- c_3 x_3 \\
\dot{x}_4 & = a_4 x_2^3 (v_2 - x_4) - c_4 x_4 
\end{aligned}
\end{equation}
Here, $a_i$, $c_j$  and $v_l$ are positive parameters.
The quantity of interest to be controlled is the \emph{concentration of intracellular zinc} ($x_1=[\mathrm{Zn}^{2+}]$).
There is one external control parameter, $\mathcal{I}$, representing the \emph{(normalized) concentration of extracelullar zinc}.

\vspace{5mm}

Thus, this network is an example of input=output (Fig.~\ref{fig:zinc}, upper left). Consistent with the theory in Subsection \ref{ssec:input_equal_output}, the algorithm detects no simple paths, simple nodes or super-simple nodes; all nodes are appendage nodes. The algorithm identifies two appendage homeostasis mechanisms, summarized in the lower-right table in Fig.~\ref{fig:zinc}: the 2-node appendage subnetwork $2 \biarrow 3$ ($\AA_1$) and null-degradation on node $4$ ($\AA_2$). These results agree with the analytic classification in \cite{antoneli2025a}.

\begin{figure}[!ht]
\centering
\includegraphics[width=\textwidth]{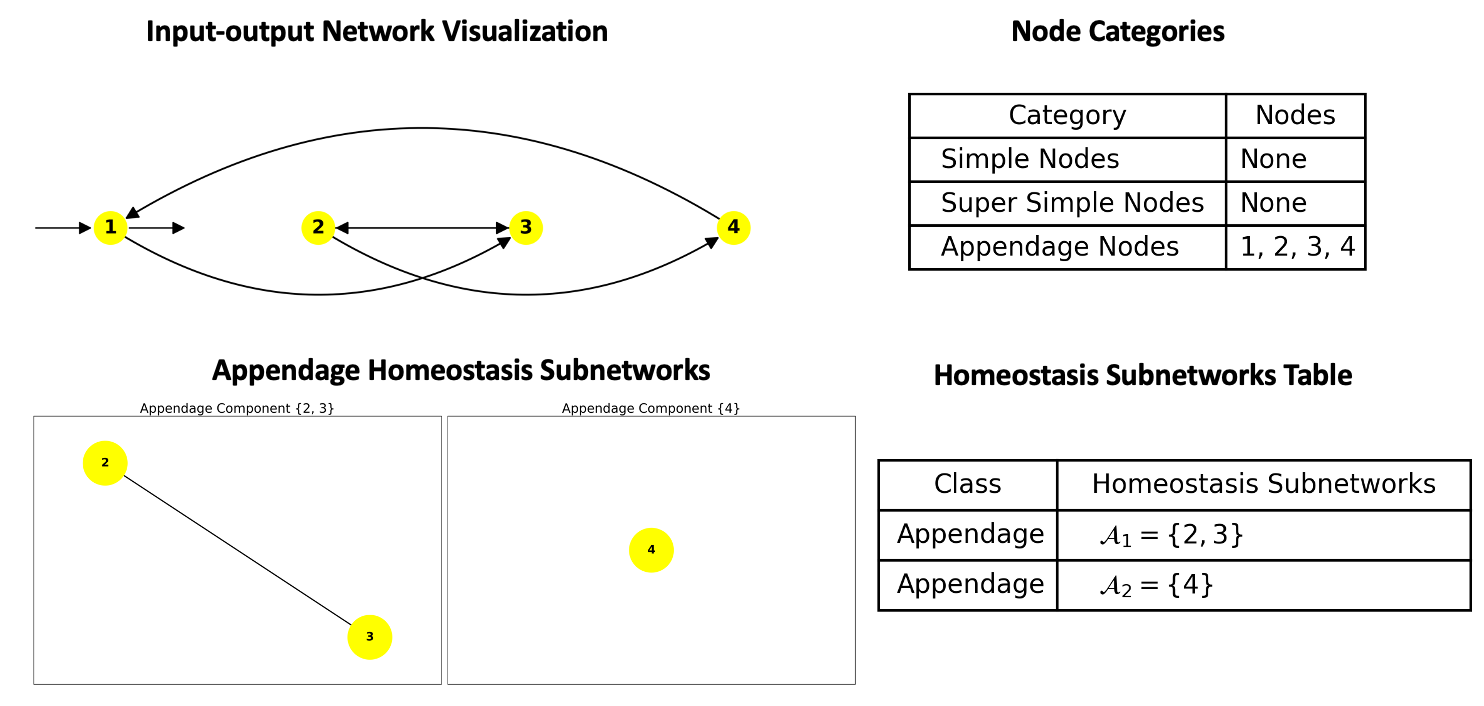}
\caption{Algorithm-generated outputs for the zinc model. This input-output network contains only appendage nodes (yellow nodes). }
\label{fig:zinc}
\end{figure}

\section{Conclusion and Outlook}
\label{sec:discussion}

Despite intensive studies on the structural properties of networks that achieve homeostasis or perfect adaptation \cite{MTELT09, AM13, TM16, F16, QD18, AL18, DQMS18, ALGBSK19,hong2025topological,khammash2021perfect,yi2000robust,xiao2018robust,cappelletti2020hidden,briat2016antithetic,araujo2023universal,qiao2019network}, the mathematical complexity of these results can make them difficult to apply in practice, thereby limiting their broader accessibility and usability.
In this paper we present a Python-based algorithm that automates the identification
of homeostasis subnetworks of an input-output network.
More specifically, given an input–output network specified solely by its connectivity structure (e.g., an edge list, adjacency matrix, or admissible ODE system) and the designation of input and output nodes, the algorithm automatically computes the relevant graph-theoretical structures and enumerates all homeostasis subnetworks.
The algorithm is able to handle several classes of input-output networks, including those with a single input parameter (with single or multiple input nodes), multiple input parameters, and cases when the input and output nodes coincide. 
For systems subject to conservation laws or other algebraic constraints \cite{jin2026,oellerich2021}, the algorithm remains applicable after reformulating the admissible system in terms of independent variables when such constraints can be identified.

We illustrate the algorithm by applying it to several biological examples, covering all the cases mentioned above.
We also included some examples of moderately large networks to showcase the algorithm's capability to handle network sizes that are beyond what can be done by hand.

The wide applicability of the algorithm, is due to a new method, developed in this paper, to reduce the analysis of networks with multiple input nodes to the case of a single input node,
by introducing an augmented single-input representation (see Theorem \ref{thm:GenH}).
This new result, together with the observation from \cite{madeira2024} that the analysis of a multiple inputs network amounts to deal with one parameter at a time, allowed us to extend the scope of an algorithm designed to treat the single input node/parameter case to other classes of networks.

The new theoretical approach introduced in this paper may be helpful in the development of another algorithm to enumerate the homeostasis patterns of an input-output network.
A \emph{homeostasis pattern} is defined as a set of nodes, in addition to the output node, that are simultaneously infinitesimally homeostatic.
It has been shown in \cite{duncan2024homeostasis,antoneli2025b} that to each homeostasis subnetwork there corresponds a unique (generic) homeostasis pattern.
Moreover, it is shown that the classification of hemostasis patterns is reduced to a combinatorial problem.
Hence, it is expected that there is a corresponding algorithm for the computation of homeostasis patterns.

\section*{Acknowledgment}

The research of FA was supported by Funda\c{c}\~ao de Amparo \`a Pes\-qui\-sa do Estado de S\~ao Paulo (FAPESP) grant 2023/04839-7.  YW was partly supported by NIH/NIDA R01DA057767. This work was also supported by the Gordon Science Research Fellows Endowment Fund, established by the Cele H. and William B. Rubin Family Fund, Inc., with additional support from the Brandeis Division of Science. 

\section*{Code Availability}

The implementation of the algorithm, together with input data files, algorithm outputs for all network examples used in this work, and documentation, is publicly available at GitHub: \url{https://github.com/Homeostasis-Classification/Homeostasis-Classification-Algorithm}. 

\section*{Competing Interest Statement}

The authors have declared no competing interest. 

\section*{Appendices}

In all examples presented in Section \ref{sec:application}, we reported the homeostasis subnetworks detected by the algorithm, denoted as $\KK_\eta$, but did not list the associated homeostasis blocks $B_\eta$. In the appendix sections below, we include the homeostasis blocks, as recorded in the output file \texttt{run\_output.txt}, for the cholesterol network in Appendix \ref{app:chol} and for the dopamine network in Appendix \ref{app:dopamine}. 

In addition, we perform in Appendix \ref{app:dopamine} an analytic classification of the homeostasis mechanisms for the dopamine network. This is obtained by computing the Jacobian matrix $J$, the homeostasis matrix $H$, and factorizing $\det(H)$ to identify all irreducible homeostasis blocks, thereby verifying the accuracy of the algorithm outputs. 

For the remaining examples, we omit listing the homeostasis blocks here; they can be found in the corresponding \texttt{run\_output.txt} files provided in the GitHub repository (see Code Availability). 


\appendix




\section{Cholesterol network}\label{app:chol}

\begin{figure}[htbp!]
    \centering
\includegraphics[width=0.7\textwidth]{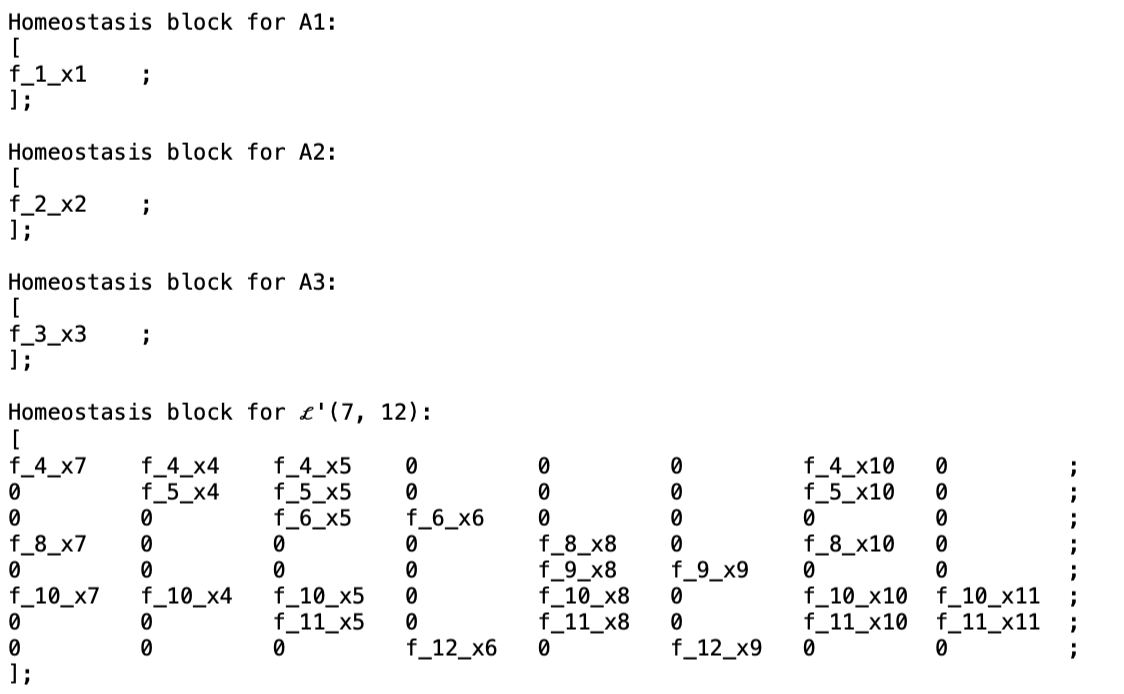}
    \caption{Algorithm output showing homeostasis blocks for the cholesterol example (Fig.~\ref{fig:Cholesterol_output12}), with the node ID-variable name mapping provided in Table \ref{tab:variable_map_cholesterol}.}
    \label{fig:cholesterole-homeostasis-conditions}
\end{figure}

The admissible system for the Cholesterol network equations proposed in \cite{pool2018integrated} is given by the following equations:

\begin{equation}\label{eq:cholesterol}
    \begin{aligned}
\frac{dx_1}{dt}  &= f_1\!\left(x_{1}, x_{12}\right)\\
\frac{dx_2}{dt}  &= f_2\!\left(x_{2}, x_{12}\right)\\
\frac{dx_3}{dt}  &= f_3\!\left(x_{1}, x_{3}\right)\\
\frac{dx_4}{dt}  &= f_4\!\left(x_{4}, x_{5}, x_{7}, x_{10}\right)\\
\frac{dx_5}{dt}  &= f_5\!\left(x_{4}, x_{5}, x_{10}\right)\\
\frac{dx_6}{dt}  &= f_6\!\left(x_{5}, x_{6}\right)\\
\frac{dx_7}{dt}  &= f_7\!\left(x_{7}, x_{8}, x_{10}\right)\\
\frac{dx_8}{dt}  &= f_8\!\left(x_{7}, x_{8}, x_{10}\right)\\
\frac{dx_9}{dt}  &= f_9\!\left(x_{8}, x_{9}\right)\\
\frac{dx_{10}}{dt} &= f_{10}\!\left(x_{4}, x_{5}, x_{7}, x_{8}, x_{10}, x_{11}\right)\\
\frac{dx_{11}}{dt} &= f_{11}\!\left(x_{2}, x_{5}, x_{8}, x_{10}, x_{11}\right)\\
\frac{dx_{12}}{dt} &= f_{12}\!\left(x_{3}, x_{6}, x_{9}, x_{12}\right)
\end{aligned}
\end{equation}
where $x_i$ denotes the state variable associated with nodes $i$, with the mapping between nodes, variables, and their biological names given in Table~\ref{tab:variable_map_cholesterol}. 

As discussed in Subsection \ref{sec:chol}, the algorithm detects three appendage and one structural homeostasis subnetworks (Fig.~\ref{fig:Cholesterol_output12}, lower-right table). Their corresponding homeostasis blocks, which are irreducible components of the homeostasis matrix $H$, are listed in Fig.~\ref{fig:cholesterole-homeostasis-conditions}.

\section{Dopamine network}\label{app:dopamine}

In this appendix section, we provide additional details for the dopamine network example discussed in Subsection~\ref{subsec:dopamine}.

\subsection{Algorithm-detected homeostasis blocks}\label{app:dopamine-block}

We consider a model of dopamine synthesis and release proposed by Best et al \cite{best2009homeostatic}, where the activity of the enzyme tyrosine hydroxylase (TH) is the \emph{external input parameter}.  lists the model variables and the corresponding nodes in the network. Variation of TH activity affects variables bh2, bh4, tyr and L-dopa (see Table \ref{tab:variable_map_dopamine}), making this an input-output network with multiple input nodes. In order to apply the algorithm, one must first convert this network $G$ into a augmented single-input-node network $G^\diamond$ by introducing a new input node TH and connecting it to bh2, bh4, tyr and L-dopa (Remark \ref{rem:multi-input} and Theorem \ref{thm:GenH}). 

As discussed in Subsection~\ref{subsec:dopamine}, the algorithm detects four structural homeostasis subnetworks (Fig.~\ref{fig: Dopamine_Example_Output}, lower-right table). Their corresponding homeostasis blocks are listed in Fig.~\ref{fig:dopamine-homeostasis-conditions}. 

\begin{figure}[htp!]
    \centering
\includegraphics[width=0.6\textwidth]{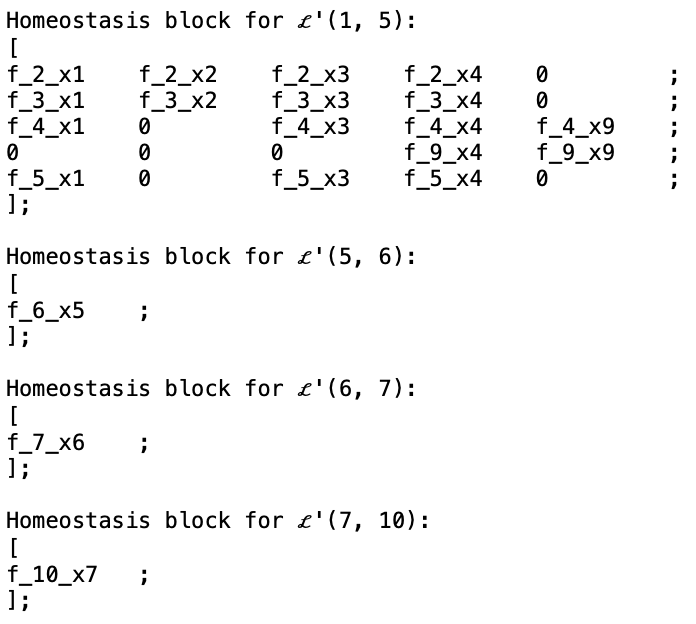}
    \caption{Algorithm output showing homeostasis blocks for the dopamine example in Fig.~\ref{fig: Dopamine_Example_Output}, with the node ID-variable name mapping provided in Table \ref{tab:variable_map_dopamine}.}
    \label{fig:dopamine-homeostasis-conditions}
\end{figure}

\subsection{Verification of algorithm outputs via direct computation}\label{app:dopa-compute}

Below, we compute the Jacobian and homeostasis matrices for the admissible system of the augmented network $G^\diamond$, determine the irreducible factors of the homeostasis matrix $H$, and compare them with the blocks detected by the algorithm. 

The admissible system of $G^\diamond$ is given by:
\begin{equation}\label{eq:DA-admissible}
\begin{aligned}
\frac{dx_1}{dt} &= f_1(x_1, \II)\\
\frac{dx_2}{dt} &= f_2(x_1, x_2, x_3, x_4, x_6, x_{10})\\
\frac{dx_3}{dt} &= f_3(x_1, x_2, x_3, x_4, x_6, x_{10})\\
\frac{dx_4}{dt} &= f_4(x_1, x_3, x_4, x_6, x_9, x_{10})\\
\frac{dx_5}{dt} &= f_5(x_1, x_3, x_4, x_5, x_6, x_{10})\\
\frac{dx_6}{dt} &= f_6(x_5, x_6, x_7, x_8, x_{10})\\
\frac{dx_7}{dt} &= f_7(x_6, x_7)\\
\frac{dx_8}{dt} &= f_8(x_8)\\
\frac{dx_9}{dt} &= f_9(x_4, x_9)\\
\frac{dx_{10}}{dt} &= f_{10}(x_7, x_8, x_{10})\\
\end{aligned}
\end{equation}
where $x_i$ denotes the state variable associated with node $i$ (see Table \ref{tab:variable_map_dopamine} for the corresponding model variables). In particular, $x_1$ denotes TH and represents the newly introduced input node influencing $x_2,\cdots, x_5$ (i.e., bh2, bh4, tyr, and L-dopa). Since node 8 is not downstream of the input node 1 and hence not part of the core network, it can be excluded when computing the homeostasis blocks.  

The Jacobian matrix of \eqref{eq:DA-admissible}, excluding $x_8$, is given by 

$$\begin{gathered}
J=\left[
\begin{array}{ccccccccc}
f_{1,x_1} & 0 & 0 & 0 & 0 & 0 & 0 & 0 & 0\\
f_{2,x_1} & f_{2,x_2} & f_{2,x_3} & f_{2,x_4} & 0 & f_{2,x_6} & 0 & 0 & f_{2,x_{10}}\\
f_{3,x_1} & f_{3,x_2} & f_{3,x_3} & f_{3,x_4} & 0 & f_{3,x_6} & 0 & 0 & f_{3,x_{10}}\\
f_{4,x_1} & 0 & f_{4,x_3} & f_{4,x_4} & 0 & f_{4,x_6} & 0 & f_{4,x_9} & f_{4,x_{10}}\\
f_{5,x_1} & 0 & f_{5,x_3} & f_{5,x_4} & f_{5,x_5} & f_{5,x_6} & 0 & 0 & f_{5,x_{10}}\\
0 & 0 & 0 & 0 & f_{6,x_5} & f_{6,x_6} & f_{6,x_7} & 0 & f_{6,x_{10}}\\
0 & 0 & 0 & 0 & 0 & f_{7,x_6} & f_{7,x_7} & 0 & 0\\
0 & 0 & 0 & f_{9,x_4} & 0 & 0 & 0 & f_{9,x_9} & 0\\
0 & 0 & 0 & 0 & 0 & 0 & f_{10,x_7} & 0 & f_{10,x_{10}}
\end{array}
\right]
\end{gathered}$$

The homeostasis matrix is obtained by removing the first row (corresponding to the input node 1) and the last column (corresponding to the output node 10):
$$
H=\begin{gathered}
\left[
\begin{array}{cccccccc}
f_{2,x_1} & f_{2,x_2} & f_{2,x_3} & f_{2,x_4} & 0 & f_{2,x_6} & 0 & 0\\
f_{3,x_1} & f_{3,x_2} & f_{3,x_3} & f_{3,x_4} & 0 & f_{3,x_6} & 0 & 0\\
f_{4,x_1} & 0 & f_{4,x_3} & f_{4,x_4} & 0 & f_{4,x_6} & 0 & f_{4,x_9}\\
f_{5,x_1} & 0 & f_{5,x_3} & f_{5,x_4} & f_{5,x_5} & f_{5,x_6} & 0 & 0\\
0 & 0 & 0 & 0 & f_{6,x_5} & f_{6,x_6} & f_{6,x_7} & 0\\
0 & 0 & 0 & 0 & 0 & f_{7,x_6} & f_{7,x_7} & 0\\
0 & 0 & 0 & f_{9,x_4} & 0 & 0 & 0 & f_{9,x_9}\\
0 & 0 & 0 & 0 & 0 & 0 & f_{10,x_7} & 0
\end{array}
\right]
\end{gathered}
$$

By cofactor expansion, we obtain 
\begin{equation}
\begin{aligned}
\det(H) &= \pm f_{6,x_5} f_{7,x_6} f_{10,x_7} \det(B_1) 
\end{aligned}
\end{equation}
where 
\[
B_1=
\left[\begin{array}{ccccc}
f_{2, x_{1}} & f_{2, x_{2}} & f_{2, x_{3}} & f_{2, x_{4}} & 0\\
f_{3, x_{1}} & f_{3, x_{2}} & f_{3, x_{3}} & f_{3, x_{4}} & 0\\
f_{4, x_{1}} & 0 & f_{4, x_{3}} & f_{4, x_{4}} & f_{4, x_{9}}\\
0 & 0 & 0 & f_{9, x_{4}} & f_{9, x_{9}}\\
f_{5, x_{1}} & 0 & f_{5, x_{3}} & f_{5, x_{4}} & 0
\end{array}\right].
\]
The four irreducible factors obtained above ($f_{6,x_5},\,f_{7,x_6},\,f_{10,x_7},\, B_1$) are the same as those produced by the algorithm (Fig.~\ref{fig:dopamine-homeostasis-conditions}). They give rise to the four structural homeostasis subnetworks identified by the algorithm.

\clearpage

\bibliography{references}

\begin{thebibliography}{59}
\providecommand{\natexlab}[1]{#1}
\providecommand{\url}[1]{{#1}}
\providecommand{\urlprefix}{URL }
\providecommand{\doi}[1]{\url{https://doi.org/#1}}
\providecommand{\eprint}[2][]{\url{#2}}
 \bibcommenthead

\bibitem[{Andrade et~al.(2022)Andrade, Madeira, and Antoneli}]{andrade2022}
Andrade PPAC, Madeira JLO, Antoneli F (2022) Homeostatic mechanisms in
  biological systems. arXiv 2202.11218:1--31. \doi{10.48550/arXiv.2202.11218}

\bibitem[{Ang and McMillen(2013)}]{AM13}
Ang J, McMillen D (2013) Physical constraints on biological integral control
  design for homeostasis and sensory adaptation. Biophys J 104(2):505--515

\bibitem[{Antoneli et~al.(2018)Antoneli, Golubitsky, and Stewart}]{AGS18}
Antoneli F, Golubitsky M, Stewart I (2018) Homeostasis in a feed forward loop
  gene regulatory network motif. J Theor Biol 445:103--109

\bibitem[{Antoneli et~al.(2025)Antoneli, Golubitsky, Jin, and
  Stewart}]{antoneli2025a}
Antoneli F, Golubitsky M, Jin J, et~al (2025) Homeostasis in input-output
  networks: Structure, classification and applications. Math Biosci 384:109435

\bibitem[{Antoneli et~al.(2026{\natexlab{a}})Antoneli, Best, Golubitsky, and
  Jin}]{antoneli2026}
Antoneli F, Best J, Golubitsky M, et~al (2026{\natexlab{a}}) Homeostasis in
  networks with same input and output nodes and metal ion regulation. In
  preparation

\bibitem[{Antoneli et~al.(2026{\natexlab{b}})Antoneli, Golubitsky, Jin, and
  Stewart}]{antoneli2025b}
Antoneli F, Golubitsky M, Jin J, et~al (2026{\natexlab{b}}) Homeostasis in gene
  regulatory networks. Int J Biomath 19:2550070.
  \doi{10.1142/S1793524525500706}

\bibitem[{Aoki et~al.(2019)Aoki, Lillacci, Gupta, Baumschlager, Schweingruber,
  and Khammash}]{ALGBSK19}
Aoki S, Lillacci G, Gupta A, et~al (2019) A universal biomolecular integral
  feedback controller for robust perfect adaptation. Nature 570:533--537

\bibitem[{Araujo and Liota(2018)}]{AL18}
Araujo R, Liota L (2018) The topological requirements for robust perfect
  adaptation in networks of any size. Nature Comm 9:1757

\bibitem[{Araujo and Liotta(2023)}]{araujo2023universal}
Araujo RP, Liotta LA (2023) Universal structures for adaptation in biochemical
  reaction networks. Nature Communications 14(1):2251

\bibitem[{Best et~al.(2009)Best, Nijhout, and Reed}]{best2009homeostatic}
Best J, Nijhout HF, Reed MC (2009) Homeostatic mechanisms in dopamine synthesis
  and release: a mathematical model. Theor Biol Med Model 6(1):1--20

\bibitem[{Briat et~al.(2016)Briat, Gupta, and Khammash}]{briat2016antithetic}
Briat C, Gupta A, Khammash M (2016) Antithetic integral feedback ensures robust
  perfect adaptation in noisy biomolecular networks. Cell systems 2(1):15--26

\bibitem[{Brualdi and Cvetkoi\'{c}(2009)}]{BC09}
Brualdi R, Cvetkoi\'{c} D (2009) A Combinatorial approach to matrix theory and
  its applications. Chapman \& Hall/CRC Press.

\bibitem[{Brualdi and Ryser(1991)}]{BR91}
Brualdi R, Ryser H (1991) Combinatorial matrix theory. Cambridge University
  Press.

\bibitem[{Cappelletti et~al.(2020)Cappelletti, Gupta, and
  Khammash}]{cappelletti2020hidden}
Cappelletti D, Gupta A, Khammash M (2020) A hidden integral structure endows
  absolute concentration robust systems with resilience to dynamical
  concentration disturbances. Journal of the Royal Society Interface 17(171)

\bibitem[{Claus and {Chavarr{\'i}a-Krauser}(2012)}]{claus2012}
Claus J, {Chavarr{\'i}a-Krauser} A (2012) Modeling regulation of zinc uptake
  via {{ZIP}} transporters in yeast and plant roots. PLOS One 7(6):e37193.
  \doi{10.1371/journal.pone.0037193}

\bibitem[{Claus et~al.(2015)Claus, Ptashnyk, Bohmann, and
  {Chavarr{\'i}a-Krauser}}]{claus2015}
Claus J, Ptashnyk M, Bohmann A, et~al (2015) Global {{Hopf}} bifurcation in the
  {{ZIP}} regulatory system. Journal of Mathematical Biology 71(4):795--816.
  \doi{10.1007/s00285-014-0836-1}

\bibitem[{Clausznitzer et~al.(2010)Clausznitzer, Oleksiuk, Løvdok, Sourjik,
  and Endres}]{clausznitzer10}
Clausznitzer D, Oleksiuk O, Løvdok L, et~al (2010) Chemotactic response and
  adaptation dynamics in \textit{{E}scherichia coli}. PLoS Comput Biol
  6(5):1--11. \doi{10.1371/journal.pcbi.1000784}

\bibitem[{Drengstig et~al.(2012)Drengstig, Ni, Thorsen, Jolma, and
  Ruoff}]{drengstig2012robust}
Drengstig T, Ni X, Thorsen K, et~al (2012) Robust adaptation and homeostasis by
  autocatalysis. The Journal of Physical Chemistry B 116(18):5355--5363

\bibitem[{Duncan et~al.(2024)Duncan, Antoneli, Best, Golubitsky, Jin, Nijhout,
  Reed, and Stewart}]{duncan2024homeostasis}
Duncan W, Antoneli F, Best J, et~al (2024) Homeostasis patterns. SIAM Journal
  on Applied Dynamical Systems 23(3):2262--2292

\bibitem[{Edgington and Tindall(2015)}]{edgington15}
Edgington MP, Tindall MJ (2015) Understanding the link between single cell and
  population scale responses of \emph{Escherichia coli} in differing ligand
  gradients. Comp Struct Biotech J 13:528--538

\bibitem[{Edgington and Tindall(2018)}]{edgington18}
Edgington MP, Tindall MJ (2018) Mathematical analysis of the
  \textit{{E}scherichia coli} chemotaxis signalling pathway. Bull Math Biol
  80(4):758--787

\bibitem[{Ferrell(2016)}]{F16}
Ferrell J (2016) Perfect and near perfect adaptation in cell signaling. Cell
  Syst 2:62--67

\bibitem[{Golubitsky and Stewart(2017)}]{GS17}
Golubitsky M, Stewart I (2017) Homeostasis, singularities and networks. J Math
  Biol 74:387--407

\bibitem[{Golubitsky and Stewart(2023)}]{gs2023}
Golubitsky M, Stewart I (2023) Dynamics and Bifurcation in Networks -
  {{Theory}} and Applications of Coupled Differential Equations. {SIAM},
  Philadelphia, PA, \doi{10.1137/1.9781611977332}

\bibitem[{Golubitsky and Wang(2020)}]{GW20}
Golubitsky M, Wang Y (2020) Infinitesimal homeostasis in three-node
  input-output networks. J Math Biol 80:1--23. \doi{10.1007/s00285-019-01457-x}

\bibitem[{Hong et~al.(2025)Hong, Moon, Hirono, and Kim}]{hong2025topological}
Hong H, Moon S, Hirono Y, et~al (2025) Topological criterion for robust perfect
  adaptation of reaction fluxes in biological networks. Iscience 28(6):112394.
  \doi{10.1016/j.isci.2025.112394}

\bibitem[{Jin and Rempala(2026)}]{jin2026}
Jin J, Rempala GA (2026) Infinitesimal homeostasis in mass-action systems.
  Journal of Mathematical Biology 92(3):35. \doi{10.1007/s00285-026-02352-y}

\bibitem[{Khammash(2021)}]{khammash2021perfect}
Khammash MH (2021) Perfect adaptation in biology. Cell Systems 12(6):509--521

\bibitem[{Lloyd(2013)}]{L13}
Lloyd A (2013) The regulation of cell size. Cell 154:1194

\bibitem[{Ma et~al.(2009)Ma, Trusina, El-Samad, Lim, and Tang}]{MTELT09}
Ma W, Trusina A, El-Samad H, et~al (2009) Defining network topologies that can
  achieve biochemical adaptation. Cell 138:760--773

\bibitem[{Madeira and Antoneli(2022)}]{madeira2022homeostasis}
Madeira JLO, Antoneli F (2022) Homeostasis in networks with multiple input
  nodes and robustness in bacterial chemotaxis. J Nonlinear Sci 32(3):37.
  \doi{10.1007/s00332-022-09793-x}

\bibitem[{Madeira and Antoneli(2024)}]{madeira2024}
Madeira JLO, Antoneli F (2024) Homeostasis in networks with multiple inputs. J
  Math Biol 89:17. \doi{10.1007/s00285-024-02117-5}

\bibitem[{Nijhout and Reed(2014)}]{NR14}
Nijhout H, Reed M (2014) Homeostasis and dynamic stability of the phenotype
  link robustness and plasticity. Integr Comp Biol 54(2):264--75

\bibitem[{Nijhout et~al.(2004)Nijhout, Reed, Budu, and Ulrich}]{NRBU04}
Nijhout H, Reed M, Budu P, et~al (2004) A mathematical model of the folate
  cycle: new insights into folate homeostasis. J Biol Chem 279:55008--55016

\bibitem[{Nijhout et~al.(2014)Nijhout, Best, and Reed}]{NBR14}
Nijhout H, Best J, Reed M (2014) Escape from homeostasis. Math Biosci
  257:104--110

\bibitem[{Nijhout et~al.(2015)Nijhout, Best, and Reed}]{NBR15}
Nijhout H, Best J, Reed M (2015) Using mathematical models to understand
  metabolism, genes and disease. BMC Biol 13:79

\bibitem[{Nijhout et~al.(2018)Nijhout, Best, and Reed}]{NBR18}
Nijhout H, Best J, Reed M (2018) Systems biology of robustness and homeostatic
  mechanisms. WIREs Syst Biol Med 11(3):e1440

\bibitem[{Nijhout et~al.(2006)Nijhout, Reed, Lam, Shane, Gregory~III, and
  Ulrich}]{nijhout2006}
Nijhout HF, Reed MC, Lam SL, et~al (2006) In silico experimentation with a
  model of hepatic mitochondrial folate metabolism. Theoretical Biology and
  Medical Modelling 3(1):40

\bibitem[{Nijhout et~al.(2017)Nijhout, Sadre-Marandi, Best, and
  Reed}]{nijhout2017systems}
Nijhout HF, Sadre-Marandi F, Best J, et~al (2017) Systems biology of phenotypic
  robustness and plasticity. Integr Comp Biol 57(2):171--184

\bibitem[{Oellerich et~al.(2021)Oellerich, Emelianenko, Liotta, and
  Araujo}]{oellerich2021}
Oellerich T, Emelianenko M, Liotta L, et~al (2021) Biological networks with
  singular jacobians: their origins and adaptation criteria. bioRxiv
  2021.03.01.433197. \doi{10.1101/2021.03.01.433197}

\bibitem[{Pool et~al.(2018)Pool, Sweby, and Tindall}]{pool2018integrated}
Pool F, Sweby PK, Tindall MJ (2018) An integrated mathematical model of
  cellular cholesterol biosynthesis and lipoprotein metabolism. Processes
  6(8):134

\bibitem[{Qian and Vecchio(2018)}]{QD18}
Qian Y, Vecchio DD (2018) Realizing 'integral control' in living cells: how to
  overcome leaky integration due to dilution? J R Soc Interface
  15(139):20170902

\bibitem[{Qiao et~al.(2019)Qiao, Zhao, Tang, Nie, and Zhang}]{qiao2019network}
Qiao L, Zhao W, Tang C, et~al (2019) Network topologies that can achieve dual
  function of adaptation and noise attenuation. Cell systems 9(3):271--285

\bibitem[{Reed et~al.(2010)Reed, Lieb, and Nijhout}]{RLN10}
Reed M, Lieb A, Nijhout H (2010) The biological significance of substrate
  inhibition: a mechanism with diverse functions. Bioessays 32(5):422--429

\bibitem[{Reed et~al.(2017)Reed, Best, Golubitsky, Stewart, and
  Nijhout}]{RBGSN17}
Reed M, Best J, Golubitsky M, et~al (2017) Analysis of homeostatic mechanisms
  in biochemical networks. Bull Math Biol 79(9):1--24

\bibitem[{Reed et~al.(2008)Reed, Thomas, Pavisic, James, Ulrich, and
  Nijhout}]{reed2008mathematical}
Reed MC, Thomas RL, Pavisic J, et~al (2008) A mathematical model of glutathione
  metabolism. Theoretical biology and medical modelling 5(1):8

\bibitem[{Sadre-Marandi et~al.(2018)Sadre-Marandi, Dahdoul, Reed, and
  Nijhout}]{sadre2018sex}
Sadre-Marandi F, Dahdoul T, Reed MC, et~al (2018) Sex differences in hepatic
  one-carbon metabolism. BMC Sys Biol 12(1):89

\bibitem[{Schneider(1977)}]{S77}
Schneider H (1977) The concepts of irreducibility and full indecomposability of
  a matrix in the works of frobenius, k\"onig and markov. Lin Alg Appl
  18:139--162

\bibitem[{Tang and McMillen(2016)}]{TM16}
Tang ZF, McMillen DR (2016) Design principles for the analysis and construction
  of robustly homeostatic biological networks. J Theor Biol 408:274--289

\bibitem[{Thom(1969)}]{thom1969}
Thom R (1969) Topological models in biology. Topology 8(3):313--335.
  \doi{10.1016/0040-9383(69)90018-4}

\bibitem[{Thom(1975)}]{thom1975}
Thom R (1975) Structural Stability and Morphogenesis. W.A. Benjamin, Inc.,
  Reading, MA.

\bibitem[{Tindall et~al.(2008{\natexlab{a}})Tindall, Porter, Maini, and
  Armitage}]{tindall08b}
Tindall MJ, Porter SL, Maini PK, et~al (2008{\natexlab{a}}) Overview of
  mathematical approaches used to model bacterial chemotaxis {II}: bacterial
  populations. Bull Math Biol 70(6):1570—1607.
  \doi{10.1007/s11538-008-9322-5}

\bibitem[{Tindall et~al.(2008{\natexlab{b}})Tindall, Porter, Maini, Gaglia, and
  Armitage}]{tindall08a}
Tindall MJ, Porter SL, Maini PK, et~al (2008{\natexlab{b}}) Overview of
  mathematical approaches used to model bacterial chemotaxis {I}: The single
  cell. Bull Math Biol 70(6):1525—1569. \doi{10.1007/s11538-008-9321-6}

\bibitem[{Vecchio et~al.(2018)Vecchio, Qian, Murray, and Sontag}]{DQMS18}
Vecchio DD, Qian Y, Murray R, et~al (2018) Future systems and control research
  in synthetic biology. Annu Rev Control 45:5--17

\bibitem[{Wang et~al.(2021)Wang, Huang, Antoneli, and Golubitsky}]{WHAG21}
Wang Y, Huang Z, Antoneli F, et~al (2021) The structure of infinitesimal
  homeostasis in input-output networks. J Math Biol 82:62.
  \doi{10.1007/s00285-021-01614-1}

\bibitem[{Wyatt et~al.(1999)Wyatt, Ritz-{De Cecco}, Czeisler, and
  Dijk}]{WRCD99}
Wyatt JK, Ritz-{De Cecco} A, Czeisler CA, et~al (1999) Circadian temperature
  and melatonin rhythms, sleep, and neurobehavioral function in humans living
  on a 20-h day. Amer J Physiol 277:1152--1163

\bibitem[{Xiao and Doyle(2018)}]{xiao2018robust}
Xiao F, Doyle JC (2018) Robust perfect adaptation in biomolecular reaction
  networks. In: 2018 IEEE conference on decision and control (CDC), IEEE, pp
  4345--4352

\bibitem[{Yi et~al.(2000)Yi, Huang, Simon, and Doyle}]{yi2000robust}
Yi TM, Huang Y, Simon MI, et~al (2000) Robust perfect adaptation in bacterial
  chemotaxis through integral feedback control. Proceedings of the National
  Academy of Sciences 97(9):4649--4653

\bibitem[{Zeeman(1977)}]{zeeman1977}
Zeeman EC (1977) Catastrophe {T}heory -- {S}elected {P}apers 1972--1977.
  Addison--Wesley, London.

\end{thebibliography}

\end{document}